\documentclass{article}
\usepackage{fullpage}
\usepackage{authblk}
\usepackage{amsmath,amssymb,amsthm,stmaryrd}
\usepackage{mdwlist}
\usepackage{thmtools}
\usepackage{thm-restate}
\usepackage{tikz}
\usepackage{changepage} 

\usetikzlibrary{arrows,fit,decorations.pathreplacing,decorations.pathmorphing,positioning,patterns,shapes}
\usetikzlibrary{trees,automata}
\tikzstyle{level 1}=[level distance=1.5cm, sibling distance=2.5cm]
\tikzstyle{level 2}=[level distance=1.5cm, sibling distance=1.5cm]
\tikzstyle{level 3}=[level distance=1.5cm, sibling distance=1cm]
\tikzstyle{level 4}=[level distance=1.5cm, sibling distance=2cm]
\usetikzlibrary{arrows,calc}

\usepackage{paralist}
\usepackage{complexity,booktabs,xspace,hyperref,wrapfig}
\usepackage{macros}
\usepackage{collect}
\usepackage{etoolbox}
\patchcmd{\quote}{\rightmargin}{\leftmargin 0em \rightmargin}{}{}
\makeatletter
\AtBeginEnvironment{quote}{\appto\@listi{\parsep\z@ \topsep\z@ \itemsep\z@}}
\makeatother

\newtheorem{definition}{Definition}
\newtheorem{remark}{Remark}
\newtheorem{example}{Example}
\newtheorem{lemma}{Lemma}
\newtheorem{theorem}{Theorem}

\definecollection{appendix-definition}
\title{Assume-Admissible Synthesis%
\footnote{Supported by the ERC starting grant inVEST (FP7-279499)}}
\author{Romain Brenguier, Jean-Fran\c cois Raskin, Ocan Sankur}
\affil{Universit\'e Libre de Bruxelles, Brussels, Belgium}

\begin{document}
\maketitle
\begin{abstract}
In this paper, we introduce a novel rule for synthesis of reactive systems,  applicable to systems made of $n$ components which have each their own objectives. 
It is based on the notion of {\em admissible} strategies. We compare our novel
rule with previous rules defined in the literature, and we show that contrary to
the previous proposals, our rule defines sets of solutions which are {\em rectangular}. This property leads to solutions which are robust and resilient. We  provide algorithms with optimal complexity and also an abstraction framework. 
\end{abstract}

\section{Introduction}
\label{section:intro}
The automatic synthesis of reactive systems has recently attracted a considerable  attention. 
The theoretical foundations of most of the contributions in this area rely on two-player zero sum games played on graphs: one player  (player~1) models the system to synthesize, and the other player (player~2) models its environment. The game is zero-sum: the objective of player~1 is to enforce the specification of the system while the objective of player~2 is the {\em negation} of this specification. This is a {\em worst-case} assumption: because the cooperation of the environment cannot be assumed, we postulate that it is {\em antagonistic}.

A fully adversarial environment is usually a bold abstraction of reality. Nevertheless, it is popular because it is {\em simple} and {\em sound}: a {\em winning strategy} against an antagonistic player is winning against {\em any} environment which pursues its own objective. But this approach may fail to find a winning strategy even if there exist solutions when the objective of the environment is taken into account. 
Also, this model is for two players only: system {\it vs} environment. 
In practice, both the system and the environment may be composed of several parts to be constructed individually or whose objectives should be considered one at a time.
It is thus crucial to take into account different players' objectives when synthesizing strategies;
accordingly, 
alternative notions have been proposed in the literature.

A first classical alternative is to {\em weaken} the winning condition of player~1 using the objective of the 
environment, requiring the system to win only when the environment meets its objective. 
This approach together with its weaknesses have been discussed in~\cite{BloemEJK14}, we will add to that later in the paper.  A second alternative is to use concepts from $n$-players {\em non}-zero sum games. This is the approach taken both by {\em assume-guarantee synthesis}~\cite{CH07} ({\sf AG}), and by {\em rational synthesis}~\cite{FismanKL10} ({\sf RS}). {\sf AG} relies on {\em secure equilibria}~\cite{KHJ06} ({\sf SE}), a refinement of Nash equilibria~\cite{nash50} ({\sf NE}). In {\sf SE}, objectives are lexicographic: players first try to maximize their own specifications, and then try to falsify the specifications of others. It is shown in~\cite{KHJ06} that {\sf SE} are those {\sf NE} which represent enforceable contracts between the two players. In {\sf RS}, the system is assumed to be monolithic and the environment is made of components that are {\em partially controllable}. In {\sf RS}, we search for a profile of strategies where the system ensures its objective and the players that model the environment are given an ``{\em acceptable}'' strategy profiles, from which it is assumed that they will not deviate.  ``Acceptable'' is formalized either by {\sf NE}, dominating strategies ({\sf Dom}), or subgame perfect equilibria ({\sf SPE}). 

\noindent\textbf{Contributions.}
As a first and central contribution, we propose a novel notion of synthesis where we take into account different players' objectives using the concept of {\em admissible} strategies~\cite{adam2008admissibility,berwanger07,BRS14}. 
For a player with objective $\phi$, a strategy $\sigma$ is {\em dominated} by $\sigma'$ if $\sigma'$ does as well as $\sigma$ w.r.t. $\phi$ against all strategies of the other players, and better for some of those strategies. A strategy $\sigma$ is {\em admissible} if it is {\em not} dominated by another strategy. 
In~\cite{berwanger07}, the admissibility notion was lifted to games played on graphs, and algorithmic questions left open were solved in~\cite{BRS14},
with the goal of \emph{model checking} the set of runs that survive the iterative elimination of dominated strategies.
Here, we use this notion to derive a meaningful notion to \emph{synthesize} systems with several players, with the following idea.
\emph{Rational} players should only play admissible strategies since dominated strategies are clearly \emph{suboptimal}.
In {\em assume-admissible synthesis} ({\sf AA}), we make the assumption that players play admissible strategies. 
Then for each player, we search for an admissible strategy that is \emph{winning} against all admissible strategies of other players. 
{\sf AA} is {\em sound}: any strategy profile that is winning against admissible strategies of other players, satisfies the objectives of all the players (Theorem~\ref{thm:aa}). 

As a second contribution, we compare the different synthesis rules. First we apply all the rules on a simple but representative example, and show the main advantages of {\sf AA} w.r.t. the other rules. 
Then we compare systematically the different approaches. We show when a solution for one rule implies a solution for another rule and we prove that, contrary to other rules, \AA yields rectangular sets of solutions (Theorem~\ref{thm:rectangular}). We argue that the rectangularity property is essential for practical applications.

As a third contribution, we provide algorithms to decide the existence of assume-admissible winning strategy profiles and prove the optimal complexity of our algorithm (Theorem~\ref{thm:aaalgo}): \PSPACE-complete for M\"uller, and \P\TIME\ for B\"uchi objectives.

As a last important contribution, we  provide an abstraction framework which allows us to define sufficient conditions to compute sets of winning assume-admissible strategies for each player in the game compositionally (Theorem~\ref{thm:abstract}). 

\noindent\textbf{Additional pointers to related works.}
We have already mentioned assume-guarantee synthesis~\cite{CH07} and rational synthesis~\cite{FismanKL10,KPV14}. Those are the closest related works to ours as they pursue the same scientific objective: to synthesis strategy profiles for non-zero sum multi-player games by taking into account the specification of each player. As those works are defined for similar formal setting, we are able to provide formal statements in the core of the paper that add elements of comparison with our work.

In~\cite{Faella09}, Faella studies several alternatives to the notion of winning strategy including the notion of admissible strategy. His work is for two-players only, and only the objective of one player is taken into account, the objective of the other player is left unspecified. Faella uses the notion of admissibility to define a notion of {\em best-effort} in synthesis while we use the notion of admissibility to take into account the 
objectives of the other players in an $n$ player setting where each player has his own objective.

The notion of admissible strategy is definable in strategy logics~\cite{ChatterjeeHP10,MogaveroMV10} and decision problems related to the \AA rule can be reduced to satisfiability queries in such logics. Nevertheless this would not lead to worst-case optimal algorithms. Based on our previous work~\cite{BRS14}, we develop in this paper worst-case optimal algorithms.

In~\cite{DammF14}, Damm and Finkbeiner use the notion of {\em dominant
  strategy} to provide a compositional semi-algorithm for the (undecidable)
  distributed synthesis problem.  So while we use the notion of admissible
  strategy, they use a notion of dominant strategy. The notion of dominant
  strategy is {\em strictly stronger}:
  every dominant strategy is admissible but an admissible strategy is not
  necessary dominant. Also, in multiplayer games with
  omega-regular objectives with complete information (as considered here),
  admissible strategies are always guaranteed to exist~\cite{berwanger07} while
  it is not the case for dominant strategies.  We will show in an example that
  the notion of dominant strategy is too strong for our purpose. 
  Also, note that the objective of Damm and Finkbeiner is different from ours: they use dominance as a mean to formalize a notion of {\em best-effort} for components of a distributed system w.r.t. their common objective, while we use admissibility to take into account the objectives of the other components  when looking for a winning strategy for one component to enforce its own objective.  Additionally,
  our formal setting is different from their setting in several respects. First, they consider zero-sum games between a
  distributed team of players (processes) against a unique environment, each
  player in the team has the same specification (the specification of the
      distributed system to synthesize) while the environment is considered as
  adversarial and so its specification is the negation of the specification of
  the system. In our case, each player has his {\em own} objective and we do not
  distinguish between protagonist and antagonist players.  Second, they consider
  distributed synthesis: each individual process has its own view of the system
  while we consider games with perfect information in which all players have a
  complete view of the system state.  Finally, let us point out
  that Damm and Finkbeiner use the term {\em admissible} for specifications and
{\em not} for strategies (as already said, they indeed consider dominant
    strategies and not admissible strategies). 
 In our case, we use the notion of
{\em admissible} strategy which is classical in game theory,
  see e.g.~\cite{FT91,adam2008admissibility}. This vocabulary mismatch is
  unfortunate but we decided to stick to the term of ``admissible strategy'' 
	which is well accepted in the literature, 
	and already used in several previous works on (multi-player) 
	games played on graphs~\cite{berwanger07,Faella09,BRS14}. 

\smallskip
\noindent{{\bf Structure of the paper.}} 
Sect.~$2$ contains definitions.
In Sect.~3, we review synthesis rules introduced in the literature and define assume-admissible synthesis. 
In Sect.~4, we consider an example; this allows us to underline some weaknesses of the previous rules. 
Sect.~5 presents a formal comparison of the different rules.
Sect.~6 contains algorithms for B\"uchi and M\"uller objectives, and Sect.~7 abstraction
techniques applied to our rule.

\section{Definitions}
\label{section:def}
  A \newdef{turn-based multiplayer arena} is a tuple
  $\A = \left\langle \Agt, (\Stat_i)_{i\in\Agt}, \sinit, (\Act_i)_{i\in\Agt}, \delta \right\rangle$
  where
  $\Agt$ is a finite set of players;
  for $i \in \Agt$, $\Stat_i$ is a finite set of player-$i$ states; 
  we let $\Stat = \biguplus_{i\in \Agt} \Stat_i$;
  $\sinit \in \Stat$ is the initial state;
  for every $i \in \Agt$, $\Act_i$ is the set of player-$i$ actions;
  we let $\Act = \bigcup_{i\in \Agt} \Act_i$;
  and $\delta \colon \Stat \times \Act \mapsto \Stat$ is the transition function.
A \newdef{run}~$\rho$ is a sequence of alternating states and actions~$\rho = s_1a_1s_2a_2\ldots \in (\Stat\cdot \Act)^\omega$ such that for all~$i\geq 1$,
$\delta(s_i,a_i)= s_{i+1}$. We write $\rho_i = s_i$, and $\act_i(\rho) = a_i$. 
A \newdef{history} is a finite prefix of a run ending in a state.
We denote by $\rho_{\leq k}$ the history $s_1a_1 \ldots s_k$;
and write $\last(\rho_{\leq k}) = s_k$, the last state of the history.
The set of states \newdef{occurring infinitely often} in a run~$\rho$ is $\inff{\rho} = \left\{ s \in \S \mid \forall j\in \mathbb{N}.\ \exists i > j, \rho_i = s \right\}$.

  A \newdef{strategy} of player $i$ is a function $\sigma_i : (\Stat^* \cdot \Stat_i) \rightarrow \Act_i$.
  A \newdef{strategy profile} for the set of players~$P \subseteq \Agt$ is a tuple of strategies, one for each player of~$P$.
  We write $-i$ for the set $\Agt \setminus \{ i \}$.
  Let $\stratset_{i}(\A)$ be the set of the  strategies of player $i$ in $\A$, written $\stratset_i$ if $\A$ is clear from context,
  and $\stratset_P$ the strategy profiles of~$P\subseteq \Agt$.

A run~$\rho$ is \emph{compatible} with strategy~$\sigma$ for player~$i$ if 
for all~$j\geq 1$, $\rho_j \in \Stat_i$ implies $\act_j(\rho) = \sigma(\rho_{\leq j})$.
It is compatible with strategy profile $\sigma_\Agt$ if it is compatible
 with each $\sigma_i$ for $i\in\Agt$.
The \emph{outcome} of a strategy profile $\sigma_\Agt$
is the unique run compatible with $\sigma_\Agt$ starting at $\sinit$,
denoted $\outcome_\A(\sigma_\Agt)$.
We write $\outcome_{\A,s}(\sigma_\Agt)$ for the outcome starting at state~$s$.
Given $\sigma_P \in \Sigma_P$
with $P \subseteq \Agt$, let $\outcome_\A(\sigma_P)$ denote the set of runs  compatible with~$\sigma_P$, and extend it to $\outcome_\A(\Sigma')$ where~$\Sigma'$ is a set of strategy profiles.
For~$E \subseteq \Stat_i \times \Act_i$, let $\strat_i(E)$ denote the set of
player-$i$ strategies~$\sigma$ whose compatible outcomes use action~$a$
from a state~$s$ only if $(s,a) \in E$.

An \emph{objective} $\phi$ is a subset of runs.
A strategy $\sigma_i$ of \player{i} is \newdef{winning} 
for objective $\phi_i$ if for all $\sigma_{-i}\in \stratset_{-i}$, $\outcome_\A(\sigma_i,\sigma_{-i}) \in \win{i}$.
A \emph{game} is an arena equipped with an objective for each player, written $\Game = \langle \A, (\win{i})_{i\in \Agt}\rangle$ where for each \player{i}, $\win{i}$ is an objective.
Given a strategy profile $\sigma_P$ for the set of players~$P$, we write $\Game,\sigma_P \models \phi$ if $\outcome_\A(\sigma_P) \subseteq \phi$.
We write $\outcome_\Game(\sigma_P) = \outcome_\A(\sigma_P)$, and $\outcome_\Game = \outcome_\Game(\Sigma)$
for $\Sigma \subseteq \stratset_i$.
For any coalition~$C \subseteq \Agt$, and objective~$\phi$, we denote by $\Win_C(\A,\phi)$
the set of states~$s$ such that there exists~$\sigma_C \in \Sigma_C$ with $\out_{\Game,s}(\sigma_C) \subseteq \phi$.

Although we prove some of our results for general objectives, 
we give algorithms for $\omega$-regular
objectives represented by Muller conditions.
A Muller condition is given by a family~$\mathcal{F}$ of
 sets of states: $\win{i} = \{ \rho \mid \inff{\rho}\in \mathcal{F}\}$.
Following~\cite{hunter07}, we assume that~$\mathcal{F}$ is given by a Boolean circuit whose inputs are~$\S$, which evaluates to true exactly on valuations
encoding subsets~$S \in \mathcal{F}$.
We also use linear temporal logic (LTL)~\cite{Pnueli77} to describe objectives.
LTL formulas are defined by $\phi := \G\phi \mid \F\phi \mid \X\phi \mid \phi \U \phi \mid \phi \W \phi \mid S$ where $S \subseteq \Stat$
(We refer to \cite{Emerson90} for the semantics.)
We consider the special case of B\"uchi objectives, given by $\G\F(B) = \{ \rho \mid B \cap \inff{\rho} \ne \varnothing\}$.
Boolean combinations of formulas $\G\F(S)$ define
Muller conditions representable by polynomial-size circuits.

In any game~$\Game$, a player~$i$ strategy $\sigma_i$ is \emph{dominated} by $\sigma'_i$ 
if for all~$\sigma_{-i} \in \Sigma_{-i}$, $\Game,\sigma_i,\sigma_{-i} \models \phi_i$ implies $\Game,\sigma'_i,\sigma_{-i} \models \phi_i$
and there exists~$\sigma_{-i} \in \Sigma_{-i}$, such that $\Game,\sigma'_i,\sigma_{-i} \models \phi_i$ and $\Game,\sigma_i,\sigma_{-i} \not\models \phi_i$,
(this is classically called \emph{weak} dominance, but we call it dominance for simplicity).
A strategy which is not dominated is \emph{admissible}.
Thus, admissible strategies are maximal, and incomparable, with respect to the dominance relation.
We write~$\adm_{i}(\Game)$ for the set of \emph{admissible} strategies in~$\Sigma_i$, and $\adm_P(\Game) = \prod_{i \in P} \adm_i(G)$ the product of the sets of
admissible strategies for~$P\subseteq \Agt$. 

Strategy~$\sigma_i$ is \emph{dominant} (\dom) if for all~$\sigma_i'$ and $\sigma_{-i}$,
  $\Game,\sigma_i',\sigma_{-i} \models \phi_i$ implies
$\Game,\sigma_i,\sigma_{-i} \models \phi_i$. The set of dominant strategies for player~$i$ 
is written $\dom_i(\Game)$. 
A \emph{Nash equilibrium} (\NE) for~$\Game$ is a strategy profile~$\sigma_\Agt$ such that for all~$i \in \Agt$, and $\sigma_i' \in \Sigma_i$,
$\Game,\sigma_{-i},\sigma_i' \models \phi_i$ implies
$\Game,\sigma_{\Agt} \models \phi_i$; thus no player can improve its outcome by deviating from the prescribed strategy.
A Nash equilibrium for $\Game$ from $s$, is a Nash equilibrium for~$\Game$ where the initial state is replaced by $s$.
A \emph{subgame-perfect equilibrium} (\SPE) for~$\Game$ is a strategy profile~$\sigma_\Agt$ such that for all histories~$h$,
$(\sigma_i \circ h)_{i\in\Agt}$ is a Nash equilibrium in $\Game$ from state $\last(h)$, where given a strategy $\sigma$, $\sigma \circ h$ denotes the strategy $\last(h)\cdot h' \mapsto \sigma(h \cdot h')$.

\section{Synthesis Rules}
\label{section:rules}
In this section, we review synthesis rules proposed in the literature, and introduce a novel one: the {\em assume-admissible} synthesis rule (\AA). Unless stated otherwise, we fix for this section a game~$\Game$, with players $\Agt=\{1,\dots,n\}$ and their objectives $\phi_1,\dots,\phi_n$.

\begin{quote}
{\bf Rule \Coop:} The objectives are \emph{achieved cooperatively} if
there is a strategy profile~$\sigma_\Agt=(\sigma_1,\sigma_2,\dots,\sigma_n)$ such that $\Game,\sigma_\Agt \models \bigwedge_{i\in\Agt} \phi_i$.
\end{quote}

This rule~\cite{MannaW81,ClarkeE81} asks for a strategy profile that {\em jointly} satisfies the objectives of all the players. 
This rule makes {\em very strong assumptions}: players fully cooperate and strictly follow their respective strategies.
This concept is {\em not robust} against deviations and postulates that the behavior of every component in the system is
 {\em controllable}. 
This weakness is well-known: see e.g.~\cite{CH07} where the rule is called {\em weak co-synthesis}.

\begin{quote}
{\bf Rule \Win.} 
The objectives are \emph{achieved adversarially} if
  there is a strategy profile~$\sigma_\Agt=(\sigma_1,\dots,\sigma_n)$ such that for all~$i \in \Agt$, 
  $\Game,\sigma_i \models \phi_i$.
\end{quote}

This rule does {\em not} require any cooperation among players: the rule asks to synthesize for each player $i$ a strategy which enforces \hisher objective $\phi_i$ against all possible strategies of the other players. 
Strategy profiles obtained by \Win are extremely {\em robust}: each player is able to ensure \hisher objective no matter how the other players behave. 
Unfortunately, this rule is often not applicable in practice: often, none of the players has a winning strategy against {\em all} possible strategies of the other players. The next rules soften this requirement by taking into account the objectives of other players.

\begin{quote}
{\bf Rule  {\sf Win-under-Hyp}}:
  Given a two-player game~$\Game$ with $\Agt=\{1,2\}$ in which player~1 has objective $\phi_1$, player~2 has objective $\phi_2$, player~1 can \emph{achieve adversarially} $\phi_1$ under {\em hypothesis} $\phi_2$, 
  if there is a strategy $\sigma_1$ for player~1 such that $\Game, \sigma_1 \models \phi_2 \rightarrow \phi_1$.
\end{quote}

The rule {\em winning under hypothesis} applies for two-player games only. Here, we consider the synthesis of a strategy for player~1 against player~2 under the hypothesis that player~$2$ behaves according to \hisher specification. 
This rule is a relaxation of the rule \Win as player~1 is {\em only} expected to win when player~2 plays so that the outcome of the game satisfies $\phi_2$. While this rule is often reasonable, it is {\em fundamentally}  plagued by the following problem: instead of trying to satisfy $\phi_1$, player~1 could try to falsify $\phi_2$, see e.g.~\cite{BloemEJK14}. This problem disappears if player~2 has a winning strategy to enforce $\phi_2$, and the rule is then safe. We come back to that later in the paper (see Lemma~1).

Chatterjee et al. in~\cite{CH07} proposed synthesis rules  inspired by {\sf Win-under-Hyp} but avoid the aforementioned problem. The rule was originally proposed in a  model with two components and a scheduler. We study here two natural extensions for $n$ players.

\smallskip
\begin{quote}
{\bf Rules ${\sf AG}^{\land}$ and ${\sf AG}^{\lor}$}:
The objectives are achieved by
  \begin{itemize*} 
  \item[($\AG^{\land}$)]  \emph{assume-guarantee-$\land$} if
    there exists a strategy profile~$\sigma_\Agt$ such that\linebreak
  \begin{inparaenum}
    \item $\Game,\sigma_\Agt \models \bigwedge_{i\in\Agt} \phi_i$,
    \item for all players $i$, $\Game,\sigma_i \models  (\bigwedge_{j\in\Agt\setminus\{i\}} \phi_j) \Rightarrow \phi_i$.
    \end{inparaenum}
 
   \item[($\AG^{\lor}$)] \emph{assume-guarantee-$\lor$}\footnote{This rule was introduced in~\cite{CDFR14}, under the name {\sf Doomsday equilibria}, as a generalization of the {\sf AG} rule of~\cite{CH07} to the case of $n$-players.}  if
  there exists a strategy profile~$\sigma_\Agt$ such that\linebreak
  \begin{inparaenum}
    \item $\Game,\sigma_\Agt \models \bigwedge_{i\in\Agt} \phi_i$,
    \item for all players $i$, $\Game,\sigma_i \models  (\bigvee_{j\in\Agt\setminus\{i\}} \phi_j) \Rightarrow \phi_i$.
    \end{inparaenum}
    
 \end{itemize*}
 \end{quote}
    
The two rules differ in the second requirement: ${\sf AG}^{\land}$ requires that player~$i$ wins whenever {\em all} the other players win, while ${\sf AG}^{\lor}$ requires player~$i$ to win whenever {\em one} of the other player wins. Clearly ${\sf AG}^{\lor}$ is stronger, and the two rules are equivalent for two-player games. 
As shown in~\cite{KHJ06}, for two-player games, a profile of strategy for ${\sf AG}^{\land}$ (or ${\sf AG}^{\lor}$) is a Nash equilibrium in a  derived game where players want, in lexicographic order, first to satisfy their own objectives, and then as a secondary objective, want to falsify the objectives of the other players. As {\sf NE}, ${\sf AG}^{\land}$ and ${\sf AG}^{\lor}$ require players to {\em synchronize} on a particular strategy profiles. As we will see, this is not the case for the new rule that we propose.

\cite{FismanKL10} and \cite{KPV14} introduce two versions of {\em rational synthesis} (\RS).  In the two cases, one of the player, say player~1, models the system while the other players model the environment. The existential version ($\RS^{\exists}$) searches for a strategy for the system, and a profile of strategies for the environment, such that the objective of the system is satisfied, and the profile for the environment is {\em stable} according to a solution concept which is either \NE, \SPE, or \dom. The universal version ($\RS^{\forall}$) searches for a strategy for the system, such that for all environment strategy profiles that are {\em stable} according to the solution concept, the objective of the system holds.  We write $\Sigma_{G,\sigma_1}^{\NE}$, resp. $\Sigma_{\Game,\sigma_1}^{\SPE}$, for the set of strategy profiles $\sigma_{-1} = (\sigma_2,\sigma_3,\dots,\sigma_n)$ 
that are \NE (resp. \SPE) equilibria in the game $\Game$ when player~1 plays $\sigma_1$, and $\Sigma_{G,\sigma_1}^{\dom}$ for the set of strategy profiles
$\sigma_{-1}$ where each strategy $\sigma_j$, $2 \leq j \leq n$, is dominant in the game $\Game$ when player~$1$ plays $\sigma_1$.

\smallskip
\begin{quote} {\bf Rules $\RS^{\exists,\forall}(\NE, \SPE, \dom)$}:
Let $\gamma \in \{ \NE,\SPE,\dom\}$, the objective is achieved by:
  \begin{itemize}
	\item[$(\RS^{\exists}(\gamma))$] {existential rational synthesis under $\gamma$}  if there is a strategy $\sigma_1$ of player~1, and a profile 
          $\sigma_{-1} \in \Sigma_{\Game,\sigma_1}^{\gamma}$, such that $\Game,\sigma_1,\sigma_{-1} \models \phi_1$.
	\item[$(\RS^{\forall}(\gamma))$]  {universal rational synthesis under $\gamma$} if there is a strategy $\sigma_1$ of player~1, such that $\Sigma_{\Game,\sigma_1}^{\gamma}\not=\emptyset$, and for all
          $\sigma_{-1} \in \Sigma_{\Game,\sigma_1}^{\gamma} $, $\Game,\sigma_1,\sigma_{-1} \models \phi_1$.
  \end{itemize}
\end{quote}
Clearly, $(\RS^{\forall}(\gamma))$ is stronger than $(\RS^{\exists}(\gamma))$ and more robust. 
As $\RS^{\exists,\forall}(\NE,\SPE)$ are derived from \NE and \SPE, they require players to synchronize on particular strategy profiles. 

\smallskip
\noindent\textit{Novel rule} We now present our novel rule based on the notion of {\em admissible strategies}.
\smallskip
\begin{quote}{\bf Rule \AA}:
The objectives are achieved by {\em assume-admissible (\AA) strategies} if
  there is a strategy profile~$\sigma_\Agt$ such that:
  \begin{inparaenum}
  \item for all $i \in \Agt$, $\sigma_i \in \adm_i(\Game)$;
  \item for all $i \in \Agt$,  $\forall \sigma_{-i}' \in \adm_{-i}(\Game).\ \Game,\sigma_{-i}',\sigma_i \models \phi_i$.
  \end{inparaenum}
\end{quote}
\smallskip

\noindent
A player-$i$ strategy satisfying conditions 1 and 2 above is called \emph{assume-admissible-winning}
(\emph{\AA-winning}). A profile of $\AA$-winning strategies is an \emph{$\AA$-winning strategy profile}.
The rule \AA  requires that each player has a strategy \emph{winning} against {\em admissible} strategies of other players.
So we assume that players do not play strategies which are {\em dominated}, which is reasonable as dominated strategies are clearly {\em suboptimal options}.

Contrary to \Coop, ${\sf AG}^{\land}$, and ${\sf AG}^{\lor}$, \AA does not require that the strategy profile is winning for each 
player. As for \Win, this is a consequence of the definition:
\begin{restatable}{theorem}{thmaawins}
  \label{thm:aa}
  For all \AA-winning strategy profile $\sigma_\Agt$, 
  $\Game, \sigma_\Agt \models \bigwedge_{i\in\Agt} \phi_i$.
\end{restatable}
The condition that \AA strategies are admissible is necessary for Thm.~\ref{thm:aa};
it does not suffice to have strategies that are winning against admissible strategies (see Appendix).

\section{Synthesis Rules at the Light of an Example}
\label{section:example}
We illustrate the synthesis rules on an example of a real-time scheduler with two tasks.
The system is composed of {\sf Sched} (player~1) and {\sf Env} (player~2).  {\sf Env} chooses the truth value for $r_1,r_2$ ($r_i$ is a request for task~$i$), and {\sf Sched} controls $q_1,q_2$ ($q_i$ means that task~$i$ has been scheduled).
Our model is a turn-based game: first, {\sf Env} chooses a value for $r_1,r_2$, then in the next round {\sf Sched}
chooses a value for $q_1,q_2$, and we repeat forever. 
The requirements for {\sf Sched} and {\sf Env} are as follows:
  \begin{inparaenum}
  	\item {\sf Sched} is not allowed to schedule the two tasks at the same time. When $r_1$ is true, then task~$1$ must be scheduled ($q_1$) {\em within} three rounds. When $r_2$ is true, task~$2$ must be scheduled ($q_2$) in {\em exactly} three rounds.
	\item Whenever {\sf Env} issues $r_i$ then it does not issue this request again before the occurrence of the grant $q_i$. {\sf Env} issues infinitely many requests $r_1$ and $r_2$.
  \end{inparaenum} 
We say that a request~$r_i$ is \emph{pending} whenever the corresponding grant has not yet been issued.
Those requirements can be expressed in LTL as follows:
\begin{itemize}
	\item $\phi_{{\sf Sched}} = \G (r_1 \rightarrow \X q_1 \lor \X \X \X q_1) \land \G(r_2 \rightarrow\X \X \X q_2) \land \G \lnot (q_1 \land q_2) .$
	\item $\phi_{{\sf Env}} = \G (r_1 \rightarrow \X (\lnot r_1 \W q_1)) \land \G (r_2 \rightarrow \X (\lnot r_2 \W q_2)) \land (\G \F r_1) \land  (\G \F r_2).$
\end{itemize}

\noindent
{{\bf A solution compatible with the previous rules in the literature.}} 
First, we note that there is no winning strategy neither for {\sf Sched}, nor for {\sf Env}. In fact, first let $\hat{\sigma_1}$ be the strategy of {\sf Sched} that never schedules any of the two tasks, i.e. leaves $q_1$ and $q_2$ constantly false. This is clearly forcing $\neg \phi_{\sf Env}$ against all strategies of {\sf Env}. Second, let $\hat{\sigma_2}$ be s.t. {\sf Env} always requests the scheduling of both task~1 and task~2, i.e. $r_1$ and $r_2$ are constantly true. 
It is easy to see that this enforces $\neg \phi_{\sf Sched}$ against any
strategy of {\sf Sched}. So, there is no solution with rule {\sf
Win}\footnote{Also, it is easy to see that {\sf Env} does not have a dominant
strategy for his specification (details are in Appendix). So, considering dominant strategies as {\em best-effort strategies} would not lead to a solution for this example. To find a solution, we need to take into account the objectives of the other players.}.
But clearly those strategies are also not compatible with the objectives of the respective players, so this leaves the possibility to apply successfully the other rules. 
We now consider a strategy profile  which is a solution for all the rules except for \AA.

Let $(\sigma_1,\sigma_2)$ be strategies for player~1 and 2 respectively, such that the outcome of $(\sigma_1,\sigma_2)$  is "{\sf Env} emits $r_1$, then {\sf Sched} emits $q_1$,  {\sf Env} emits $r_2$, then {\sf Sched} waits one round and emits $q_2$, and repeat." 
If a deviation from this {\em exact} execution is observed, then the two players switch to strategies $\hat{\sigma_1}$  and $\hat{\sigma_2}$ respectively, i.e. to the strategies that falsify the specification of the other players. The reader can now convince himself/herself that $(\sigma_1,\sigma_2)$ is a solution for {\Coop}, \AG and ${\sf RS}^{\exists}(\NE,\SPE,\dom)$. Furthermore, we claim that $\sigma_1$ is a solution for {\sf Win-under-Hyp} and ${\sf RS}^{\forall}(\NE,\SPE,\dom)$.\footnote{The interested reader can find a detailed explanation for all those claims in the appendix. }
But, assume now that {\sf Env} is a device driver which requests the scheduling
of tasks by the scheduler of the kernel of an OS. 
Clearly $(\sigma_1,\sigma_2)$, which is compatible with all the previous rules (but {\sf Win}), makes little sense in this context. On the other hand, $\phi_{\sf Sched}$ and $\phi_{\sf Env}$ are natural specifications for such a system. So, there is clearly room for other synthesis rules!

\noindent
{{\bf Solutions provided by {\sf AA}, our novel rule.}}
%
For {\sf Env}, we claim that the set of {\em admissible strategies}, noted  ${\sf Adm}(\phi_{{\sf Env}})$, are exactly those that $(i)$ do not emit a new request before the previous one has been acknowledged, and $(ii)$ do always eventually emit a (new) request when the previous one has been granted. Indeed as we have seen above, {\sf Env} and {\sf Sched} can cooperate to satisfy $\phi_{{\sf Sched}} \land  \phi_{{\sf Env}}$, so any strategy of {\sf Env} which would imply the falsification of  $\phi_{{\sf Env}}$ is dominated and so it is not admissible. Also, we have seen that {\sf Env} does not have a winning strategy for $\phi_{{\sf Env}}$, so {\sf Env} cannot do better.

Now, let us consider the following strategy for {\sf Sched}.
 $(i)$ if pending requests $r_1$ and $r_2$ were made one round ago, then grant $q_1$;
if pending requests $r_1$ and $r_2$ were made three rounds ago, then behave arbitrarily (it is no more possible to satisfy the specification);
 $(ii)$ if pending request $r_2$ was made three rounds ago, but not~$r_1$, then grant~$q_2$; $(iii)$ if pending $r_1$ was made three rounds ago, but not~$r_2$, then grant~$q_1$. 
 We claim that this strategy is {\em admissible} and while it is {\em not} winning against {\em all} possible strategies of {\sf Env}, it is {\em winning} against {\em all admissible} strategies of {\sf Env}. 
 So, this strategy enforces $\phi_{\sf Sched}$ against all reasonable strategies of {\sf Env} w.r.t. to \hisher own objective $\phi_{\sf Env}$. In fact, there is a whole set of such strategies for {\sf Sched}, noted ${\sf WinAdm_{\sf Sched}}$. Similarly, there is a whole set of strategies for {\sf Env} which are both {\em admissible} and winning against the {\em admissible} strategies of {\sf Sched}, noted ${\sf WinAdm_{\sf Env}}$. We prove in the next section that the solutions to \AA are {\em rectangular sets}: they are exactly the solutions in ${\sf WinAdm_{\sf Sched}} \times {\sf WinAdm_{\sf Env}}$. This ensures that \AA leads to {\em resilient} solutions: players do not need to synchronize with the other players on a particular strategy profile but they can arbitrarily choose inside their sets of strategies that are admissible and winning against the admissible strategies of the other players.

\section{Comparison of Synthesis Rules}
\label{section:comparison}
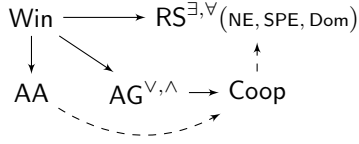
\begin{figure}[h]
  \centering
  \begin{tikzpicture}
    \node at (0,-1) (AA) {$\AA$};
    \node at (1.5,-1) (AG) {$\AG^{\lor,\land}$};
    \node at (3,-1) (C) {$\Coop$};
    \node at (3,0) (RS) {$\textsf{RS}^{\exists,\forall}($\scriptsize$\NE,\SPE,\dom$\normalsize$)$};
    \node at (0,0) (W) {\textsf{Win}};
    \draw[-latex',dashed] (AA) edge[bend right] (C);
    \draw[-latex',dashed] (C) edge (RS);
    \draw[-latex'] (W) edge (AA);
    \draw[-latex'] (W) edge (AG);
    \draw[-latex'] (W) edge (RS);
    \draw[-latex'] (AG) edge (C);
  \end{tikzpicture}
  \caption{Comparison of synthesis rules.
  }
  \label{fig:implications}
\end{figure}
In this section, we compare the synthesis rules to understand which ones
yield solutions more often, and to assess their robustness.
 Some relations are easy to establish; for instance,
rules $\Win,\AG^\lor,\AG^{\land},\AA$ imply \textsf{Coop} by definition
(and Thm.~\ref{thm:aa}).
We summarize
the implication relations between the rules in Fig.~\ref{fig:implications}.
We present the 
rules $\AG^\lor,\AG^\land$,
 and the variants of $\textsf{RS}^\cdot(\cdot)$
in one group, respectively. 
A dashed arrow from \textsf{A} to \textsf{B} means
that rule~\textsf{A} implies \emph{some} rule in \textsf{B};
while a plain arrow means that \textsf{A} implies \emph{all}
rules in \textsf{B}
(\textit{e.g.} \AA{} implies $\AG^\land$ but not~$\AG^\lor$;
while \textsf{Win} implies both rules.)
An absence of path means that \textsf{A} does not imply any variant of \textsf{B}.
Thus the figure explains which approaches yield solutions more often,
by abstracting away the precise variants.
The following theorem states the correctness of our diagram;
the appendix contains detailed proofs between all pairs of  rules.
\begin{theorem}\label{thm:implications}
  The implication relations of Fig.~\ref{fig:implications} hold.
\end{theorem}

In the controller synthesis framework using two-player games between a controller and its environment,
some works advocate the use of environment objectives which the environment can
guarantee against any controller~\cite{CHJ08}.
Under this assumption, {\sf Win-under-Hyp} implies \AA:

\begin{restatable}{lemma}{lmwuhaa}
  \label{lem:aadv->aa}  
  Let $\Game=\langle \A,\phi_1,\phi_2\rangle$ be a two-player game.
  If player~$2$ has a winning strategy for~$\phi_2$ and {\sf Win-under-Hyp} has a solution, 
  then \AA has a solution.
\end{restatable}

We now consider the \emph{robustness} of the profiles synthesized using the above rules.
An \AA-winning strategy profile~$\sigma_\Agt$ is robust in the following sense:
The set of \AA-winning profiles is \emph{rectangular}, \textit{i.e.} \emph{any} combination of \AA-winning strategies independently chosen
for each player, is an \AA-winning profile. Second, 
if one replaces \emph{any} subset of strategies in \AA-winning profile $\sigma_\Agt$ by arbitrary admissible strategies, the objectives of all the other players still hold.
Formally, a \emph{rectangular set} of strategy profiles is a set that is a Cartesian product of sets of strategies, given for each player.
A synthesis rule is \emph{rectangular} if the set of strategy profiles satisfying the rule is rectangular.
The $\RS$ rules require a specific definition since \player{1} has a particular role:
we say that $\RS^{\forall,\exists}(\gamma)$ is rectangular if given for any strategy $\sigma_1$ witnessing the rule, the set of strategy profiles $(\sigma_{2},\dots,\sigma_n) \in \Sigma^{\gamma}_{\Game,\sigma_1}$ such that $\Game, \sigma_1,\dots,\sigma_n \models \phi_1$ is rectangular.
We show that apart from~$\AA$, only $\Win$ and $\RS^\forall(\dom)$ are rectangular among
the other rules:

\begin{restatable}{theorem}{thmrectangular}
  \label{thm:rectangular}
  We have 
  \begin{inparaenum}
  \item
    Rule \AA{} is rectangular; and
    for all games~$\Game$, 
    \AA-winning strategy profile~$\sigma_P$, coalition $C\subseteq \Agt$, if $\sigma'_C \in \adm_C(\Game)$,
    then
    $\Game,\sigma_{-C},\sigma'_C \models \bigwedge_{i \in -C} \phi_i$.
  \item 
    The rules $\Win$ and $\RS^\forall(\dom)$ are rectangular;
    the rules $\Coop$, $\AG^\lor$, $\AG^\land$, \linebreak $\RS^\exists(\NE,\SPE,\dom)$, and $\RS^\forall(\NE,\SPE)$ are not rectangular.
  \end{inparaenum}
\end{restatable}


\section{Algorithm for Assume-Admissible Synthesis}
\label{section:algo-aa}
\begin{wrapfigure}{r}{0.42\textwidth}
	\vspace{0cm}
  \begin{center}
    \begin{tikzpicture}[xscale=2,every node/.style={draw,circle,minimum
      size=3mm,node distance=2cm}]
      \node at (0,0) (A) {$s_1$};
      \node[rectangle,minimum size=6mm,right of=A] (B) {$s_2$};
      \node[right of=B] (C) {$s_3$};
      \path[-latex']
        ($(A)+(-0.4,0)$) edge (A)
        (A) edge[bend left] (B)
        (B) edge[bend left] (A)
        (A) edge[loop above] (A)
        (C) edge[loop above] (C)
        (B) edge (C);
      \draw[-latex'] (A) .. controls +(0.5,0.8) and +(-0.5,0.8) .. (C);
    \end{tikzpicture}
  \end{center}
  \caption{Game~$G$ with two players~$\Agt=\{1,2\}$. Player~$1$ controls the
  round states, and has objective $\G\F s_2$, and player~$2$ controls the square state and has objective $\G\F s_1$.}
  \label{fig:running-example}
	\vspace{-0.2cm}
\end{wrapfigure}
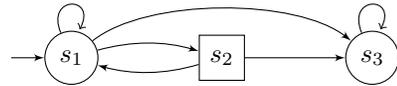
In this section, we recall the characterization of the outcomes of admissible strategy profiles given in~\cite{BRS14},
and we derive algorithms for the  \AA rule.
We use the game of Fig.~\ref{fig:running-example} as a running example for this section. Clearly, none of the players of this game has a winning strategy for his own objective when not taking into account the objective of the other player, but, as we will see, both players have an admissible and winning strategy against the {\em admissible} strategies of the other player, and so the \AA rule applies. 

The notion of \emph{value} associated to the states  of a game plays an important role in the characterization of  admissible strategies and their outcomes
\cite{berwanger07,BRS14}. Fix a game~$\Game$.
A state~$s$ has value~$1$ for player~$i$, written $\val_i(s)=1$, if player~$i$ has a winning strategy from
$s$; 
$\val_i(s) = -1$ if for all strategy profiles $\sigma_\Agt \in \Sigma_\Agt$, $\outcome_{\Game,s}(\sigma_\Agt)$ does not satisfy $\phi_i$;
and otherwise $\val_i(s) = 0$.
A \emph{\player{j} decreases its own value in history $h$} if there is a position $k$ such that $\val_j(h_{k}) > \val_j(h_{k+1})$ and $h_k \in \Stat_j$.
We proved in \cite{BRS14}, that admissible strategies do not decrease their own values.
Let us call such strategies \emph{value-preserving}.
In fact, if the current state has value~$1$, there is a winning strategy
which stays within the winning region;
if the value is~$0$, 
then although other players may force the play into states of value $-1$, a good strategy for player~$i$ will not do this by itself.

\begin{lemma}[{\cite[Lem.~1]{BRS14}}]
  \label{lemma:nodecrease}
  For all games~$\Game$, players~$i$, and histories~$\rho$, 
  if $\last(\rho) \in \Stat_i$ and $\sigma_i \in \Adm_i$ then $\val_i(\delta(\last(\rho),\sigma_i(\rho))) = \val_i(\last(\rho))$.
\end{lemma}

For player~$i$, let us define the sets $V_{i,x} = \{ s \mid \val_i(s) = x\}$ for~$x \in \{-1,0,1\}$, which partition~$\Stat$. We define the set of \emph{value-preserving edges} for player~$i$ as
\(E_i = \{ (s,a) \in \Stat \times \Act \mid s\in \S_i \Rightarrow \val_i(\delta(s,a)) = \val_i(s)\}.\)
Observe that value-preserving strategies for player~$i$ are exactly those respecting~$E_i$.
In our running example of Fig.~\ref{fig:running-example}, it should be clear that any strategy that chooses a transition that goes to $s_3$ is {\em not} admissible nor for Player~1 neither for Player 2, as by making this choice both players are condemned to lose their own objective while their other choices leave a chance to win; so the choice of going to $s_3$ would decrease their own value. So, we can already conclude that Player~2 always chooses $s_2 \mapsto s_1$, his only admissible strategy.

Not all value-preserving strategies are admissible:
for M\"uller objectives, staying inside the winning region does not imply the objective.
Moreover, in states of value~$0$, admissible strategies must visit states where other players can ``help'' satisfy the objective.
Formally, \emph{help states} for player~$i$ are other players' states with value~$0$ and at least two different successors of value $0$ or $1$.
Let
\(
H_i = \{ s \in \S \setminus \S_i  \mid \val_i(s) = 0 \land \exists s' \neq s''.\ s' \in \delta(s,\Act)
\land s'' \in \delta(s,\Act) \land \val_i(s')\ge 0 \land \val_i(s'')\ge 0\}.
\)
Given this,
the following lemma, adapted from~\cite{BRS14}, characterizes the outcomes of admissible strategies.
We denote by $\G (E_i)$ the set of runs that respect  $E_i$, \textit{i.e.}
 $\G (\bigvee_{(s,a)\in E_i} s \land \X(\delta(s,a)))$.

\begin{restatable}{lemma}{lmphii}
  \label{lemma:adm-outcomes}
  For all games~$\Game$, and players~$i$,
  $\Out_\Game \cap \Phi_i = \Out_\Game(\Adm_i,\Sigma_{-i})$,
  where
  $\Phi_i = \G(E_i) \land (\G\F(V_{i,1}) \Rightarrow \phi_i) \land (\G\F(V_{i,0}) \Rightarrow \phi_i \lor \G\F(H_i))$.
\end{restatable}

In our running example of Fig.~\ref{fig:running-example}, a strategy of Player~1 which, after some point, always chooses $s_1 \mapsto s_1$ is dominated by strategies that chose infinitely often $s_1 \mapsto s_2$. This is a corollary of the lemma above. Indeed, while all those strategies only visit states with value $0$ (and so do not decrease the value for Player~1), the strategy that always chooses $s_1 \mapsto s_1$ has an outcome which is loosing for Player~1 while the other strategies are compatible with outcomes that are winning for Player~1. So, outcome of admissible strategies for Player~1 that always visit states with values $0$, also visits $s_2$ infinitely often. Using the fact that strategies are value-preserving and the last observation, we can now conclude that both players have (admissible) winning strategies against the admissible strategies of the other players. For instance when Player~1 always chooses to play $s_1 \mapsto s_2$, he wins against the admissible strategies of Player~2.

Note that $\Phi_i$ can be decomposed into a safety condition $S_i = \G (E_i)$ and a \emph{prefix independent} condition $M_i = (\G\F(V_{i,1}) \Rightarrow \phi_i) \land (\G\F(V_{i,0}) \Rightarrow (\phi_i \lor \G\F(H_i))$ which can be expressed by a M\"uller condition described by a circuit of polynomial size.

For player~$i$, we let
\(\Omega_i = \Out_\Game(\Adm_i) \land (\Out_\Game(\Adm_{-i}) \Rightarrow \phi_i),\)
which describes the outcomes of admissible strategies of player~$i$, 
which satisfy objective~$\phi_i$ under the hypothesis that they are compatible with other players' admissible strategies.
In fact, it follows from~\cite{BRS14} that $\Omega_i$ captures the \emph{outcomes} of \AA-winning strategies for player~$i$.
\begin{restatable}{lemma}{lmomegai}
	\label{lemma:aa-omega}
  A \player{i} strategy is \AA-winning iff it is winning for objective $\Omega_i$.
\end{restatable}

Objective~$\Omega_i$ is not directly expressible as a M\"uller condition, since $\Phi_i$ and $\bigwedge_j \Phi_j$ contain safety parts. 
Nevertheless, the information whether $\G(E_i)$, or $\G(\cup_{j \neq i} E_j)$ has been violated 
can be encoded in the state space. Formally, for each player~$i$, we define game~$\Game_i'$ by taking the product of~$\Game$ with $\{\top,0,\bot\}$; that is, the states are 
$\Stat\times\{\top,0,\bot\}$, and the initial state $(\sinit,0)$. The transitions are defined
as for~$\Game$ for the first component; while from state $(s,0)$, any action~$a$ outside~$E_i$ leads to 
$(\delta(s,a),\bot)$, and any action~$a$ outside~$E_j$, $j \neq i$, leads to $(\delta(s,a),\top)$.
The second component 
is absorbing at $\bot,\top$. We now rewrite the condition~$\Omega_i$ for~$\Game_i'$ as
$\Omega'_i = \left(\G\F (\Stat\times\{0\} ) \land M'_i \land  (\wedge_{j\ne i}
M'_{j}  \Rightarrow \phi'_i) \right) \lor \left(\G\F (\Stat\times\{\top\}) \land
M'_i\right)$,
where~$M_i'$ is the set of runs of~$\Game_i'$ whose projections to~$\Game$
are in~$M_i$, and similarly for~$\phi_i'$.

Now, checking \AA-synthesis is reduced to 
solving games with M\"uller conditions. Moreover, we also obtain a polynomial-time
algorithm when all objectives are B\"uchi conditions, by showing that~$\Omega_i'$ 
is expressible by a parity condition with four colors.

\begin{restatable}{theorem}{thmaaalgo}
	\label{thm:aaalgo}
  \AA-synthesis in multiplayer games is \PSPACE-complete, and \P-complete for B\"uchi objectives.
  Player~$i$ wins for objective~$\Omega_i$ in~$\Game$ iff he wins for objective~$\Omega_i'$ in~$\Game_i'$. 
\end{restatable}

\section{Abstraction}
\label{section:abstraction}
We present abstraction techniques to compute assume-admissible strategy profiles
following the \emph{abstract interpretation} framework~\cite{cc77}; see~\cite{hmmr00} for games.
Abstraction is a crucial feature for scalability in practice, and we show here that
the $\AA$ rule is amenable to abstraction techniques.
The problem is not directly reducible to computing \AA-winning strategies in abstract games obtained as \textit{e.g.} in \cite{agj04}; in fact, it can be easily seen that the set of admissible strategies of an abstract game is incomparable with those of the concrete game in general.

\smallskip
\noindent\textbf{Overview.}
Informally, to compute an \AA-winning strategy for player~$i$, we construct an abstract game 
$\AbsA_i'$ with objective~$\underline{\Omega}_i'$
s.t. winning strategies of player~$i$ in~$\AbsA_i'$ map to \AA-winning strategies in~$\Game$.
To define~$\AbsA'$, we re-visit the steps of the algorithm of
Section~\ref{section:algo-aa} by defining
approximations computed on the abstract state space.
More precisely, 
we show how to compute under- and over-approximations of the sets ${V}_{x,k}$,
   namely $\underline{V}_{x,k}$ and $\overline{V}_{x,k}$, using fixpoint
computations on the abstract state space only. We then use these sets to define
approximations of the value preserving edges ($\underline{E}_k$
and~$\overline{E}_k$) and those of the help states ($\underline{H}_k$
and~$\overline{H}_k$). These are then combined to define objective
$\underline{\Omega'}_k$ s.t. if player~$k$ wins the abstract game 
for $\underline{\Omega'}_k$, then he wins the original game for $\Omega_k'$, and
thus has an \AA-winning strategy.

\smallskip
\noindent\textbf{Abstract Games.}
Consider $\Game = \langle \A, (\phi_i)_{i\in\Agt}\rangle$  with $ \A = \left\langle \Agt, (\Stat_\agt)_{\agt\in\Agt}, \sinit, (\Act_\agt)_{\agt\in\Agt}, \delta \right\rangle$
where each $\phi_\agt$ is a M\"uller objective 
given by a family of sets of states $(\mathcal{F}_\agt)_{\agt \in \Agt}$.
Let~$\Stat^\abs = \biguplus_{i \in \Agt}\Stat^\abs_i$ denote a finite set, 
namely the \emph{abstract state space}.
A \emph{concretization function} $\gamma \colon \Stat^\abs\mapsto 2^\Stat$ is a function such that:
\begin{inparaenum}
  \item the abstract states partitions the state space: $\biguplus_{s^\abs \in \Stat^\abs} \gamma(s^\abs) = \Stat$,
  \item it is compatible with players' states:
    for all players~$i$
    and $s^\abs \in \Stat_i^\abs$, $\gamma(s^\abs) \subseteq \Stat_i$.
  \end{inparaenum}%
We define the corresponding \emph{abstraction} function ${\alpha} : \Stat \rightarrow \Stat^\abs$ where
$\alpha(s)$ is the unique state $s^\abs$ s.t. $s \in \gamma(s^\abs)$.
We also extend  $\alpha,\gamma$ naturally to sets of states; 
and to histories, by replacing each element of the sequence by its image.

We further assume that $\gamma$ is \emph{compatible} with all objectives~$\mathcal{F}_\agt$ in the sense 
that the abstraction of a set~$S$ is sufficient to determine whether $S \in \mathcal{F}_\agt$:
for all $i \in \Agt$, for all $S,S' \subseteq \Stat$ with $\alpha(S) = \alpha(S')$, we have
     $S \in \mathcal{F}_i \Leftrightarrow S' \in \mathcal{F}_i$.
If the objective $\phi_i$ is given by a circuit, then the circuit for the corresponding abstract objective~$\phi_i^\abs$ is obtained by replacing each input on state~$s$ by $\alpha(s)$.
We thus have $\rho \in \phi_i$ if, and only if, $\alpha(\rho) \in \phi^\abs_i$.

The \emph{abstract transition relation $\Delta^{\abs}$} induced by~$\gamma$ is defined by:
\((s^\abs,a,t^\abs) \in \Delta^{\abs} \Leftrightarrow  \exists s \in \gamma(s^\abs), \exists t \in \gamma(t^\abs), t = \delta(s,a).\)
We write $\post_\Delta(s^\abs,a) = \{ t^\abs \in \Stat^\abs \mid \Delta(s^\abs,a,t^\abs)\}$,
and $\post_\Delta(s^\abs,\Act) =\cup_{a \in \Act} \post_\Delta(s^\abs,a)$.
For each coalition~$C\subseteq \Agt$, we define a game in which players $C$ play together against coalition $-C$;
and the former resolves  non-determinism in $\Delta^\abs$.
Intuitively, the winning region for $C$ in this abstract game will be an
over-approximation of the original winning region.
Given~$C$, the \emph{abstract arena $\AbsA^{C}$} is
$\langle \{ C , -C \}, (\Stat_C,\Stat_{-C}), \alpha(\sinit), \linebreak (\Act_C, \Act_{-C}), \delta^{\abs,C} \rangle$,
where
$\Stat_C = \left(\bigcup_{i\in C} \Stat^\abs_i\right) \cup \left(\bigcup_{i\in \Agt} \Stat^\abs_i \times \Act_i \right)$, 
  $\Stat_{-C} = \bigcup_{i\not\in C} \Stat^\abs_i$;
and
$\Act_C = \left(\bigcup_{i\in C} \Act_i\right) \cup \Stat^\abs$ 
  and $\Act_{-C} = \bigcup_{i\in -C} \Act_i$.
The relation $\delta^{\abs,C}$ is given by: 
if $s^\abs\in \Stat^\abs$, then $\delta^{\abs,C}(s^\abs,a) = (s^\abs,a)$.
If $(s^\abs,a) \in \Stat^\abs \times \Act$ and $t^\abs \in \Stat^\abs$ satisfies 
$(s^\abs,a,t^\abs)\in \Delta^\abs$ 
then $\delta^{\abs,C}((s^\abs,a), t^\abs) = t^\abs$;
while for $(s^\abs,a,t^\abs)\not\in \Delta^\abs$, the play leads to an arbitrarily chosen 
state~$u^\abs$ with $\Delta(s^\abs,a,u^\abs)$.
Thus, from states $(s^\abs,a)$, coalition $C$ chooses a successor $t^\abs$.

We extend $\gamma$ to histories of~$\AbsA^{C}$ by first removing states of $(\Stat^\abs_i\times \Act_i)$;
and extend $\alpha$ by inserting these intermediate states.
Given a strategy~$\sigma$ of player~$k$ in $\AbsA^{C}$, we define its \emph{concretization} as the strategy $\gamma(\sigma)$ of~$\Game$ that, at any history~$h$ of~$\Game$, plays $\gamma(\sigma)(h) = \sigma(\alpha(h))$.
We write $\Win_D(\AbsA^{C},\phi^{\abs}_k)$ for the states of $\Stat^\abs$ from which the coalition $D$ has a winning strategy in $\AbsA^{C}$ for objective $\phi^\abs_k$, with $D \in \{C,-C\}$.
Informally, it is easier for coalition~$C$ to achieve an objective in~$\AbsA^{C}$ than in~$\Game$,
that is, $\Win_C(\AbsA^C,\phi^\abs_k)$ over-approximates $\Win_C(\A,\phi_k)$:
 
\begin{restatable}{lemma}{lmwinabs}
  \label{lem:winning-abstract}
  If the coalition $C$ has a winning strategy for objective $\phi_k$ in $\Game$ from $s$ then it has a winning strategy for $\phi_k^\abs$ in $\AbsA^{C}$ from $\alpha(s)$.
\end{restatable}


\smallskip
\noindent\textbf{Value-Preserving Strategies.}
\label{section:approximate-admissible}
We now provide under- and over-approximations for value-preserving
strategies for a given player.
We start by computing approximations $\underline{V}_{k,x}$
and~$\overline{V}_{k,x}$ of the sets
$V_{k,x}$, and then use these
to obtain approximations of the value-preserving edges $E_k$.
Fix a game $\Game$, and a player $k$. 
Let us define the \emph{controllable predecessors} for player~$k$ as
 $\cpre_{\AbsA^{\Agt\setminus\{k\}},k}(X)
= \{ s^\abs \in \Stat^\abs_k \mid \exists a \in \Act_k, \post_{\Delta}(s^\abs,a) \subseteq X\}
\cup \{ s^\abs \in \Stat^\abs_{\Agt\setminus \{k\}} \mid \forall a \in \Act_{-k}, \post_{\Delta}(s^\abs,a)
\subseteq X\}$.
We let
\[
\begin{array}{l r}
\overline{V}_{k,1} = \Win_{\{k\}}(\AbsA^{\{k\}}, \phi^{\abs}_k),
&
\overline{V}_{k,-1} = \Win_\emptyset(\AbsA^{\emptyset}, \lnot \phi^{\abs}_k),\\ 
\overline{V}_{k,0} = \Win_{\Agt\setminus \{k\}}(\AbsA^{\Agt\setminus\{k\}}, \lnot \phi^{\abs,k})
     \cap \Win_{\Agt}(\AbsA^{\Agt},\phi^{\abs}_k), & \\
\underline{V}_{k,1} = \Win_{\{k\}}(\AbsA^{\Agt\setminus\{k\}}, \phi^{\abs}_k), 
&
\underline{V}_{k,-1} =  \Win_\emptyset(\AbsA^{\Agt}, \lnot\phi^{\abs}_k) \\ 
\underline{V}_{k,0} = \nu X. \big(\cpre_{\AbsA^{\Agt\setminus\{k\}},k}(X \cup \underline{V}_{k,1} \cup \underline{V}_{k,-1}) \cap F\big),\\
\qquad \text{where }
F = \Win_{\Agt\setminus \{k\}}(\AbsA^{\{k\}}, \lnot \phi^{\abs}_k) 
     \cap \Win_{\Agt}(\AbsA^{\emptyset},\phi^{\abs}_k).
\end{array}
\]
The last definition uses the $\nu X.f(X)$ operator which is the greatest fixpoint of~$f$.
These sets define approximations of the sets $V_{k,x}$. Informally, this follows from the
fact that to define \textit{e.g.} $\overline{V}_{k,1}$, we use the game~$\AbsA^{\{k\}}$, 
where player~$k$ resolves itself the non-determinism, and thus has more power than in~$\Game$.
In contrast, for~$\underline{V}_{k,1}$, we solve~$\AbsA^{\Agt\setminus\{k\}}$ where the adversary 
resolves non-determinism. We state these properties formally:

\begin{restatable}{lemma}{lmapprv}
  \label{lemma:approximate-v}
  For all players~$k$ and $x \in  \{-1,0,1\}$, 
  $\gamma(\underline{V}_{k,x}) \subseteq V_{k,x} \subseteq \gamma(\overline{V}_{k,x})$.
\end{restatable}

We thus have $\cup_x \gamma(\overline{V}_{k,x}) = \S$
(as $\cup_x V_{k,x} = \S$)
but this is not the case for $\underline{V}_{k,x}$;
so let us define $\underline{V} = \cup_{j \in \{-1,0,1\}} \underline{V}_{k,j}$.
We now define approximations of~$E_k$ based on the above sets.
\[
  \begin{array}{l}
    \overline{E}_k = \{ (s^\abs,a) \in \Stat^\abs \times \Act \mid
    s^\abs \in \Stat^\abs_k \Rightarrow
    \exists x,
    s^\abs \in \overline{V}_{k,x}, \post_{\Delta}(s^\abs,a) \cap \cup_{l \geq x} \overline{V}_{k,l}
    \neq \emptyset \},\\
    \underline{E}_k = \{ (s^\abs,a) \in \Stat^\abs \times \Act \mid
    s^\abs \in \Stat^\abs_k \Rightarrow
    \exists x,
    s^\abs \in \underline{V}_{k,x}, \post_{\Delta}(s^\abs,a) \subseteq  \cup_{l \geq x} \underline{V}_{k,l}\}
     \hfill\cup \{ (s^\abs,a) \mid s^\abs \not \in \underline{V}\}.
 \end{array}
\]

Intuitively, $\overline{E}_k$ is an over-approximation of~$E_k$,
and $\underline{E}_k$ under-approximates $E_k$ when restricted to states
in $\underline{V}$ (notice that $\underline{E}_k$ contains all actions from states outside~$\underline{V}$). 
In fact, our under-approximation will be valid only inside $\underline{V}$; but
we will require the initial state to be in this set, 
and make sure the play stays within~$\underline{V}$.
We show that sets $\underline{E}_k$ and~$\overline{E}_k$ provide approximations of
value-preserving strategies.
\begin{restatable}{lemma}{lmadmstrat}
  \label{lemma:under-admissible-strat}
  For all games~$\Game$, and players~$k$, $\strat_k(E_k) \subseteq \gamma(\strat_k(\overline{E}_k))$,
  and if $\sinit \in \gamma(\underline{V})$,
  then $\emptyset \neq \gamma(\strat_k(\underline{E}_k)) \subseteq \strat_k(E_k)$.
\end{restatable}

\smallskip
\noindent\textbf{Abstract Synthesis of \AA-winning strategies.}
We  now  describe the computation of \AA-winning strategies 
in abstract games.
Consider game $\Game$ and assume sets $\underline{E}_i, \overline{E}_i$ are computed for all players~$i$. 
Roughly, to compute a strategy for player~$k$, we will constrain him to play
only edges from~$\underline{E}_k$, while other players~$j$ will play
in~$\overline{E}_j$. By Lemma~\ref{lemma:under-admissible-strat}, 
any strategy of player~$k$ maps to value-preserving strategies in the original
game, and all value-preserving strategies for other players are still present.
We now formalize this idea, incorporating the help states in the
abstraction.

We fix a player~$k$. We construct an abstract game in which winning for
player~$k$ implies that player~$k$ has an effective \AA-winning strategy in~$\Game$.
We also define $\calA_{k}' = \langle \{\{k\},-k\}, ({\Stat'}^\abs_k,
{\Stat'}^\abs_{-k} \cup {\Stat'}^\abs \times \Act), \alpha(\sinit), (\Act_k,
\Act_{-k}), \delta_{\calA^k}\rangle$,
where ${\Stat'}^\abs = \Stat^\abs \times \{\bot,0,\top\}$;
thus we modify $\AbsA^{\Agt\setminus\{k\}}$ by taking the product of the state space with $\{\top,0,\bot\}$.
Intuitively, as in Section~\ref{section:algo-aa}, initially the second component is $0$, meaning that no player has violated 
the value-preserving edges. The component becomes~$\bot$ whenever player~$k$ plays an action outside of~$\underline{E}_k$;
and~$\top$ if another player~$j$ plays outside~$\overline{E}_j$.
We extend~$\gamma$ to~$\AbsA_k'$ by $\gamma( (s^\abs,x)) = \gamma(s^\abs) \times \{x\}$,
and extend it to histories of~$\AbsA_k'$ by first removing the intermediate states
${\Stat'}^\abs \times \Act$.
We thus see $\calA_k'$ as an abstraction of~$\A'$ of Section~\ref{section:algo-aa}.

In order to define the objective of $\calA_k'$, let us first define approximations of the help states $H_k$,
where we write $\Delta(s^\abs,\Act,t^\abs)$ to mean
$\exists a \in \Act, \Delta(s^\abs,a,t^\abs)$.
\[\begin{array}{ll}
  \overline{H}_k & = \{ s^\abs \in \overline{V}_{k,0} \setminus \Stat^\abs_k \mid \exists t^\abs, u^\abs \in \overline{V}_{k,0} \cup \overline{V}_{k,1}.\ \Delta(s^\abs,\Act,t^\abs) \land \Delta(s^\abs,\Act,u^\abs)\}\\%
  \underline{H}_k & = \{
  s^\abs \in \underline{V}_{k,0} \setminus \Stat^\abs_k \mid \exists a \neq b \in \Act, \post_\Delta(s^\abs,a) \cap \post_\Delta(s^\abs,b) = \emptyset, \\
  &\hspace{3cm}
  \post_\Delta(s^\abs,a) \cup \post_\Delta(s^\abs,b)\subseteq \underline{V}_{k,0} \cup \underline{V}_{k,1} 
  \}.
\end{array}
\]
We define the following approximations of the objectives $M_k'$ and $\Omega_k'$ in $\calA'_k$. 
\[
\begin{array}{l}
  \underline{M}'_k = ( \G\F(\overline{V}_{k,1}) \Rightarrow \phi_k^\abs)
  \land \left( \G\F(\overline{V}_{k,0}) \Rightarrow (\phi^\abs_k \lor \G\F(\underline{H}_k)) \right),\\
  \overline{M}'_k = (\G\F(\underline{V}_{k,1}) \Rightarrow \phi_k^\abs)
  \land \left( \G\F(\underline{V}_{k,0}) \Rightarrow (\phi_k^\abs \lor \G\F(\overline{H}_k)) \right),\\
  \underline{\Omega'}_k = \left(\G\F(\Stat^\abs\times\{0\}) \land \underline{M'}_k \land \left( \bigwedge_{j \neq k} \overline{M'}_j \Rightarrow {\phi}^\abs_k \right)\right)
  \lor \left(\G\F(\Stat^\abs\times\{\top\}) \land \underline{M'}_k\right).\\
\end{array}
\]

\begin{restatable}{theorem}{thmabstract}
  \label{thm:abstract}
  For all games~$\Game$, and players~$k$,
  if $\sinit \in \underline{V}$, and player~$k$
  has a winning strategy in~$\AbsA'_k$ for objective~$\underline{\Omega}_k'$, 
  then he has a winning strategy in~$\Game_k'$ for~$\Omega_k$;
  and thus a $\AA$-winning strategy in~$\Game$.
\end{restatable}
Now, if Theorem~\ref{thm:abstract} succeeds to find an \AA-winning strategy for
each player~$k$, then the resulting strategy profile is \AA-winning.

\section{Conclusion}

In this paper, we have introduced a novel synthesis rule, called the {\em assume admissible synthesis}, for the synthesis of strategies in non-zero sum $n$ players games played on graphs with omega-regular objectives.  We use the notion of admissible strategy, a classical concept from game theory, to take into account the objectives of the other players when looking for winning strategy of one player. We have compared our approach with other approaches such as assume guarantee synthesis and rational synthesis that target the similar scientific objectives.
We have developed worst-case optimal algorithms to handle our synthesis rule as well as dedicated abstraction techniques.
As future works, we plan to develop a tool prototype to support our assume admissible synthesis rule.

\bibliographystyle{plain}
\bibliography{biblio}

\newpage

\appendix

\section{Complements on Synthesis Rules (Section~\ref{section:rules})}
We give the proof of Theorem~\ref{thm:aa}.

\thmaawins*
\begin{proof}
  Let $\sigma_\Agt$ be a strategy profile witness of \AA.
  Let $i$ be a player, we have that $\sigma_{-i} \in \Adm_{-i}(\Game)$, because by condition 1, for all $j \ne i$, $\sigma_j \in \Adm_j(\Game)$.
  Then by condition $2$ we have that $\Game, \sigma_\Agt \models \phi_i$.
  Since this is true for all players $i$, we have that $\Game, \sigma_\Agt \models \bigwedge_{i\in\Agt} \phi_i$.
\end{proof}

The following example shows that 
\AA-winning strategies must be admissible themselves
for Thm.~\ref{thm:aa} to hold.

 \begin{figure}[ht]
 \centering{
    \begin{tikzpicture}[scale=2]
      \draw (0,0) node[draw,minimum size=8mm,circle] (A1) {$s_1$};
      \draw (0,-1) node[draw,minimum size=8mm,circle,inner sep=0pt] (A2) {$\phi_1,\phi_2$};
      \draw (1,0) node[draw,minimum size=8mm] (A3) {$s_2$};
      \draw (1,-1) node[draw,minimum size=8mm,circle] (A4) {$s_3$};
      \draw (2,0) node[draw,minimum size=8mm,circle] (A5) {$\bot$};
      \draw (2,-1) node[draw,minimum size=8mm,circle] (A6) {$\phi_1$};
      \draw[-latex',very thick] (A1) -- (A2);
      \draw[-latex',dashed] (A1) -- (A3);
      \draw[-latex',very thick] (A3) -- (A4);    
      \draw[-latex',dashed] (A3) -- (A5);
      \draw[-latex', very thick] (A4) -- (A2);    
      \draw[-latex', very thick] (A4) -- (A6);
    \end{tikzpicture}
  }
  \caption{Illustration of the necessity of Condition 1 in the definition of Assume-Admissible Synthesis.
    Player~1 controls circles and player~2 squares.}
  \label{fig:adm-cond2}
\end{figure}
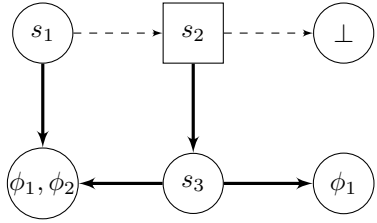

\begin{example}
  In \AA, the profile of strategy must be composed of admissible strategies only. This is necessary otherwise the players may make assumptions on each other which are not satisfied. 
  This is illustrated by the example of \figurename~\ref{fig:adm-cond2} in which two players have reachability objectives $\phi_1$ and~$\phi_2$ respectively.
  The states labeled by $\phi_i$ are the target states for the corresponding objective.
  Admissible strategies are shown in thick edges. Now, the player~$2$ strategy that chooses the dashed edge from $s_2$ satisfies Condition 2 of \AA, since $s_2$ is not reachable under admissible strategies of player~$1$.
  Similarly, the player~$1$ strategy that chooses the dashed edge from~$s_1$ satisfies Condition 2 of \AA since the thick edges lead back to a state satisfying~$\phi_1$. But then the resulting profile is such that none of the two players win.
\end{example}

We explain why the pair of strategies $(\sigma_1,\sigma_2)$ is a solution for the synthesis rules {\sf Coop}, ${\sf AG}$,  ${\sf RS}^{\exists}(\NE,\SPE)$:
  \begin{itemize}
  	\item Rule {\sf Coop}: the outcome of $(\sigma_1,\sigma_2)$  is "{\sf Env} emits $r_1$, then {\sf Sched} emits $q_1$,  {\sf Env} emits $r_2$, then {\sf Sched} waits one round and emits $q_2$, and repeat." This outcome is a model for both $\phi_{\sf Sched}$ and $\phi_{\sf Env}$.
	\item Rules ${\sf AG}^{\land,\lor}$: the two rules coincide as we have only two players. When both players follow $(\sigma_1,\sigma_2)$, we know that the outcome is a model for both $\phi_{\sf Sched}$ and $\phi_{\sf Env}$. We must in addition verify that $\outcome(\sigma_1) \models \phi_{\sf Env} \rightarrow \phi_{\sf Sched}$ and that $\outcome(\sigma_2) \models \phi_{\sf Sched} \rightarrow \phi_{\sf Env}$. This is the case as either the outcome of the game between the two players is "{\sf Env} emits $r_1$, then {\sf Sched} emits $q_1$,  {\sf Env} emits $r_2$, then {\sf Sched} waits one round and emits $q_2$, and repeat", or it is not but then player~1 switches to strategy $\hat{\sigma}_1$ and this forces outcomes that falsify $\phi_{\sf Env}$. The argument is similar to justify that $\outcome(\sigma_2) \models \phi_{\sf Sched} \rightarrow \phi_{\sf Env}$.
	\item Rules ${\sf RS}^{\exists}(\NE,\SPE,\dom)$: for the existential version of rational synthesis, we consider that {\sf Sched} (player~1) plays the role of the system. We must then prove that $(\sigma_1,\sigma_2)$ is such that $\phi_{\sf Sched}$ is satisfied and $\sigma_2$ is stable for \NE,\SPE, and \dom. This is the case because $\outcome(\sigma_1,\sigma_2)$ is a model for $\phi_{\sf Sched}$, so it satisfies the specification of the scheduler. Furthermore, it satisfies the specification of the environment, and so the environment has no incentive to deviate: $\sigma_2$ is stable for \NE. It is also stable for \SPE as any deviation of the environment triggers strategy $\hat{\sigma}_1$ for the scheduler and $\hat{\sigma}_2$ for the environment, as a consequence, there is no incentive for environment to deviate from $\hat{\sigma}_2$ as $\outcome(\hat{\sigma}_1) \models \neg \phi_{\sf Env}$.  It is stable for \dom because $\sigma_2$ is winning under $\sigma_1$.
  \end{itemize}
  
We explain why the pair of strategies $\sigma_1$ is a solution for the rules {\sf Win-under-Hyp} and ${\sf RS}^{\forall}(\NE,\SPE)$:
  \begin{itemize}
  	\item Rule {\sf Win-under-Hyp}: we have seen above that $\outcome(\sigma_1) \models \phi_{\sf Env} \rightarrow \phi_{\sf Sched}$ and so $\sigma_1$ is a solution for {\sf Win-under-Hyp}.
	\item Rules ${\sf RS}^{\forall}(\NE,\SPE,\dom)$: against $\sigma_1$, the only strategies for the environment that are stable for \NE, \SPE, \dom are those that make sure that the outcome is equal to  "{\sf Env} emits $r_1$, then {\sf Sched} emits $q_1$,  {\sf Env} emits $r_2$, then {\sf Sched} waits one round and emits $q_2$, and repeat." Indeed, any strategy that would deviate from that outcome would trigger the execution of $\hat{\sigma}_1$ for the scheduler and the outcome would then be losing for the environment making a deviation toward $\sigma_2$ profitable for instance.
  \end{itemize}

  \subparagraph{Absence of dominant strategies}
 Let us now convince the reader that Env has no dominant strategy for his
 specification. Indeed, it is easy to see that if there are dominant strategies
 then all admissible strategies must be dominant. Now, take two strategies in ${\sf Adm}(\phi_{{\sf Env}})$ that behaves the same with the exception that the first strategy when a request is granted ($q_i$ is true), directly issue another request $r_i$, for $i \in \{1,2\}$, while the second strategy waits one additional turn to issue the new request $r_i$. It is easy to see that none of the two strategies dominate the other one, as the {\sf Sched} can adapt his own strategy to make the first strategy win and the second loose, and the other way around. 

\section{Comparison of the Synthesis Rules (Section~\ref{section:comparison})}
In this section, we detail the comparison of the synthesis notions we consider.
We adopt the following notation: an implication $\textsf{A} \Rightarrow \textsf{B}$ between two
notions \textsf{A} and \textsf{B} means that \textsf{B} has a solution whenever \textsf{A} has one.
Note that this does not always imply inclusion between the witnessing strategy profiles.
The proof of Theorem~\ref{thm:implications}, will be decomposed in several lemmas.

\subsection{Proof of Theorem~\ref{thm:implications}}
\begin{remark}
  We have $\RS^\exists(\SPE) \Rightarrow \RS^\exists(\NE)$ and $\RS^\forall(\NE) \Rightarrow \RS^\forall(\SPE)$
  because any subgame perfect equilibrium is also a Nash equilibrium.
  Moreover, in the definition of the rules $\RS$, the conditions for $\RS^\forall$ are stronger than for $\RS^\exists$, so $\RS^\forall(\SPE) \Rightarrow \RS^\exists(\SPE)$, $\RS^\forall(\NE) \Rightarrow \RS^\exists(\NE)$ and $\RS^\forall(\dom) \Rightarrow \RS^\exists(\dom)$.
\end{remark}

\begin{lemma}\label{lem:win->aa}
  $\Win \Rightarrow \AA \Rightarrow \Coop \Rightarrow \RS^\exists(\SPE)$ and $\Win \Rightarrow \AG^{\lor} \Rightarrow \AG^\land \Rightarrow \Coop$,
\end{lemma}
\begin{proof}
  \fbox{$\Win \Rightarrow \AA$} 
  This holds because winning strategies are always admissible~\cite{berwanger07}, therefore a profile witness of \textsf{Win} satisfies condition 1 and 2 of the definition of assume-admissible.

  \fbox{$\AA \Rightarrow \Coop$} 
  This holds by Theorem~\ref{thm:aa}.


  \fbox{$\Coop \Rightarrow \RS^\exists(\SPE)$}
  Note that in order for $\RS$ to make sense we must have ${\sf sys} \in \Agt$.
  Assume $\Coop$ has a solution and let $\sigma_\Agt$ be a profile of strategy such that for all \player{i}, $\sigma_\Agt \models \phi_i$.

  We define a strategy profile $\sigma'_i$, that follows the path $\rho=\Out_\Game(\sigma_i)$ when possible (that is: if $h$ is a prefix of $\rho$ then play $\act_{|h|}(\rho)$) and if not follows a subgame perfect equilibrium:
  that is, we select for each state~$s$ a subgame perfect equilibrium $\sigma^s_\Agt$, there always exist one for Borel games (so in particular for Muller games)~\cite[Theorem~3.15]{ummels2010stochastic}; 
  then if $h$ is not a prefix of $\rho$, let $j$ be the last index such that $h_{\le j} = \rho_{\le j}$ and we define $\sigma'_\Agt(h) = \sigma^{h_{j+1}}_\Agt(h_{\ge j+1})$.
  
  Let $h$ be a history.
  If $h$ is a prefix of $\rho$ then the objective of each player is satisfied by following $\sigma'_i \circ h$ so none of them can gain by changing its strategy, therefore it is a Nash equilibrium from~$\last(h)$.
  If $h$ is not a prefix of $\rho$ then by definition of $\sigma'_i$, players follow a subgame-perfect equilibrium since $h$ deviated from $\rho$, so in particular $\sigma'_i \circ h$ is a Nash equilibrium from $\last(h)$.
  Moreover the objective of the system is satisfied.
  Therefore $\sigma_\Agt$ is a solution to $\RS(\SPE)$.

  \fbox{$\Win \Rightarrow \AG^\lor$} 
  Let $\sigma_\Agt$ such that for each \player{i}, $\sigma_i$ is winning for $\phi_i$.
  The first condition in the definition of $\AG^\lor$ is satisfied because for all \player{i}, $\Out_\Game(\sigma_\Agt) $ satisfies $\phi_i$.
  The second condition is satisfied because for all strategy $\sigma'_{-i}$, we have that $\Out_\Game(\sigma_i,\sigma'_{-i}) $ satisfies $\phi_i$, so in particular it satisfies $(\bigvee_{j\in \Agt\setminus\{i\}} \phi_j \Rightarrow \phi_i)$.
  Hence $\sigma_\Agt$ is a solution for $\AG^\lor$.

  \fbox{$\AG^{\lor} \Rightarrow \AG^\land$}
  This holds because the second condition in the definition of these rules is stronger for $\AG^\lor$.

  \fbox{$\AG^\land \Rightarrow \Coop$}
  This implication holds simply because of the condition 1 in the definition of assume-guarantee, which corresponds to the definition of Cooperative synthesis.

\end{proof}

\begin{lemma}
  For all $\gamma\in \{\NE, \SPE, \dom\}$, $\Win \Rightarrow \RS^\forall(\gamma)$, $\RS^\forall(\dom) \not\Rightarrow \RS^\forall(\NE)$ and
  $\RS^\exists(\dom) \Rightarrow \RS^\exists(\SPE)$.
\end{lemma}
\begin{proof}
  \fbox{$\Win \Rightarrow \RS^\forall(\gamma)$}
  Let $\sigma_\Agt$ be a strategy profile such that for each \player{i}, $\sigma_i$ is winning for $\phi_i$.

  We first show that $\Sigma^\gamma_{\Game,\sigma_1}$ is not empty.
  For $\gamma \in \{\NE,\SPE\}$ this is because there always exist a subgame perfect equilibrium for Borel games (so in particular for Muller games)~\cite[Theorem~3.15]{ummels2010stochastic} and a subgame perfect equilibrium is a Nash equilibrium.
  For $\gamma=\dom$, note that by definition of dominant strategies, winning strategies are dominant, so $\Sigma^\dom_{\Game,\sigma_1}$ contains at least $\sigma_{-1}$.

  Let $\sigma'_{-1}$ be a strategy profile for $\Agt \setminus \{1\}$.
  Since $\sigma_1$ is a winning we have that $\Game,\sigma_{1},\sigma'_{-1} \models \phi_1$.
  Therefore $\sigma_1$ is a solution for $\RS^\forall(\gamma)$.

  \medskip

  \fbox{$\RS^\exists(\dom) \Rightarrow \RS^\exists(\SPE)$} 
  Let $\sigma_\Agt$ be a witness for $\RS^\exists(\dom)$.
  We define a strategy profile $\sigma'_\Agt$ such that $\sigma'_i$ follows $\sigma_i$ on all histories compatible with $\sigma_i$ (that is if $h$ prefix of $\rho\in \Out_\Game(\sigma_i)$ then $\sigma'_i(h) = \sigma_i(h)$) and outside of these histories follows a subgame perfect equilibria: there always exist one for Borel games (so in particular for Muller games)~\cite[Theorem~3.15]{ummels2010stochastic}.
  
  By definition of $\sigma'_\Agt$, the outcome $\Out_\Game(\sigma'_\Agt)$ is the same than $\Out_\Game(\sigma_\Agt)$.
  Because $\sigma_\Agt$ is a witness for $\RS^\forall(\dom)$, this outcome is winning for \player{1}.

  It remains to show that $\sigma'_{-1}$ is a subgame perfect equilibria.
  Let $h$ be a history, $i$ be a player different from \player{1}, and $\sigma''_i$ be a strategy for \player{i}.
  We show that from $h$ player $i$ does not improve by switching from $\sigma'_i$ to another strategy $\sigma''_i$, which will show that $\sigma'_\Agt \circ h$ is a Nash equilibrium from $h$.

  If $h$ is compatible with $\sigma_i$ then $\sigma'_i$ coincide with $\sigma_i$ from this history, so $\Out_\Game(\sigma'_i,\sigma_{-i}) = \Out_\Game(\sigma_\Agt)$.
  Since $\sigma_i$ is a dominant strategy, if $\Game,\sigma''_i, \sigma_{-i} \models \phi_i$ then $\Game,\sigma_i,\sigma_{-i} \models \phi_i$ and therefore this implies that $\Out_\Game(\sigma'_i,\sigma_{-i})$ satisfy $\phi_i$.
  This means that $i$ does not improve by switching from $\sigma'_i$ to $\sigma''_i$.
 
  If $h$ is not compatible with $\sigma_i$, then $\sigma'_i$ plays according to a subgame-perfect equilibria since the first deviation. 
  In particular, this strategy is a Nash equilibrium from~$h$.

  This shows that $\sigma'_{-1}$ is a subgame perfect equilibrium and has $\sigma'_\Agt \models \phi_1$, this is a witness for $\RS^\exists(\SPE)$.

  \medskip

  \fbox{$\RS^\forall(\dom) \not\Rightarrow \RS^\forall(\NE)$} 
  Consider the example given in \figurename~\ref{fig:rs-dom-not-rs-spe}.
  The strategy $r$ for \player{2} is dominant and any strategy of \player{3} is dominant.
  The outcome of these strategies always go to the bottom state where $\phi_{\sf sys}$ is satisfied.
  Therefore there is a solution to $\RS^\forall(\dom)$.

  However, we show that there is no solution to $\RS^\forall(\NE)$.
  Consider the strategy profile $(\cdot,l,b)$, this is a Nash equilibrium (even
	a subgame Nash equilibrium) since no player can improve his/her strategy.
  Note that \player{1} is losing for that profile, hence no strategy of \player{1} can ensure that it will win for all Nash equilibria.
  \begin{figure}[ht]
  \begin{center}
  \begin{tikzpicture}[xscale=2]
    \draw (0,0) node[draw,circle,minimum size=8mm] (A) {$s_1$};
    \draw (1,0) node[draw,minimum size=8mm] (B) {$s_2$};
    \draw (0,-1.5) node[draw,minimum size=8mm] (C) {$\phi_1,\phi_2,\phi_3$};
    \draw (2,0) node[draw,minimum size=8mm] (D) {$\phi_2,\phi_3$};
    \draw (2,-1) node[draw,minimum size=8mm] (E) {$\phi_3$};
    \draw[-latex'] (-0.5,0) -- (A);
    \draw[-latex'] (A) -- node[above] {$l$} (B);
    \draw[-latex'] (A) -- node[right] {$r$} (C);
    \draw[-latex'] (B) -- node[above] {$b$} (D);
    \draw[-latex'] (B) -- node[below] {$a$} (E);
  \end{tikzpicture}
  \end{center}
  \caption{Example showing that $\RS^\forall(\dom) \not\Rightarrow \RS^\forall(\NE)$.
    Player~2 controls circle states, \player{3} square states and \player{1} does not control any state.}
  \label{fig:rs-dom-not-rs-spe}
  \end{figure}
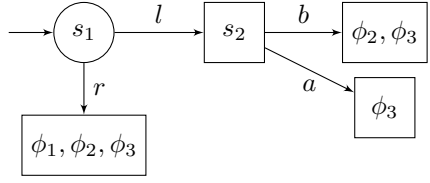
\end{proof}


In the example of Section~\ref{section:example}, we saw that more strategy profiles satisfied the assume-guarantee condition compared to assume-admissibility,
including undesirable strategy profiles. 
We show that the rule $\AG^\land$ is indeed more often satisfied than $\AA$;
while the rules $\AG^\lor$, and $\AA$ are incomparable.

\begin{lemma}\label{lem:aa->ag}
  We have $\AG^\land \not\Rightarrow \AA$; $\AG^\lor \not \Rightarrow \AA$; $\AA \not\Rightarrow \AG^\land$ and $\AA \not \Rightarrow \AG^\lor$.
\end{lemma}
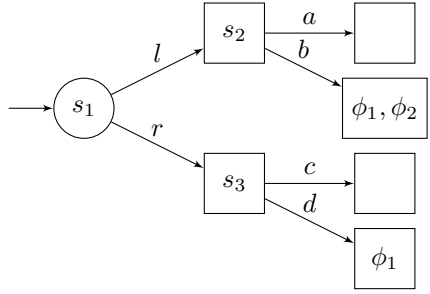
\begin{figure}[ht]
  \begin{center}
  \begin{tikzpicture}[xscale=2]
    \draw (0,0) node[draw,circle,minimum size=8mm] (A) {$s_1$};
    \draw (1,1) node[draw,minimum size=8mm] (B) {$s_2$};
    \draw (1,-1) node[draw,minimum size=8mm] (C) {$s_3$};
    \draw (2,-2) node[draw,minimum size=8mm] (D) {$\phi_1$};
    \draw (2,-1) node[draw,minimum size=8mm] (E) {};
    \draw (2,0) node[draw,minimum size=8mm] (F) {$\phi_1,\phi_2$};
    \draw (2,1) node[draw,minimum size=8mm] (G) {};
    \draw[-latex'] (-0.5,0) -- (A);
    \draw[-latex'] (A) -- node[above] {$l$} (B);
    \draw[-latex'] (A) -- node[above] {$r$} (C);
    \draw[-latex'] (B) -- node[above] {$b$} (F);
    \draw[-latex'] (B) -- node[above] {$a$} (G);
    \draw[-latex'] (C) -- node[above] {$d$} (D);
    \draw[-latex'] (C) -- node[above] {$c$} (E);
  \end{tikzpicture}
  \end{center}
  \caption{Example showing that $\AG\not\Rightarrow \AA$. 
    Player~1 controls circle states and \player{2} square states.}
  \label{fig:ag-not-aa}
  \end{figure}
  
\begin{proof}
  \fbox{$\AG^\land \not\Rightarrow \AA$ and $\AG^\lor \not \Rightarrow \AA$}  
  Consider the game represented in \figurename~\ref{fig:ag-not-aa}.
  In this example, we have $\adm_1 = \Sigma_1$. 
  Therefore \player{2} has no winning strategy against all admissible strategies of $\adm_2$ (in particular the strategy of \player{1} that plays $r$, makes \player{2} lose).
  So \AA fails. However, we do have $\AG^\land$ by the profile $\sigma_1 \colon s_1 \mapsto l, \sigma_2 \colon s_2 \mapsto b, s_3 \mapsto c$.
  This profile also satisfies $\AG^\lor$ which is equivalent to $\AG^\land$ for two player games.
  \medskip

  \fbox{$\AA \not\Rightarrow \AG^\land$}
  Consider the example of \figurename~\ref{fig:aa-not-ag-and}.
  The profile where \player{1} and \player{2} plays to the right is assume-admissible.
  However there is no solution to assume-guarantee synthesis: if \player{1} and \player{2} change their strategies to go to the state labeled $\phi_1, \phi_2$, then the condition $\calG,\sigma_3 \models (\phi_1 \land \phi_2) \Rightarrow \phi_3$ is not satisfied.
  \begin{figure}[ht]
    \begin{center}
      \begin{tikzpicture}[xscale=2]
        \draw (0,0) node[draw,rectangle,minimum size=5mm] (A) {$s_1$};
        \draw (1,0) node[draw,rectangle,minimum size=5mm] (B) {$\phi_1,\phi_2,\phi_3$};
        \draw (0,-1) node[draw,circle,minimum size=5mm] (E) {$s_2$};
        \draw (1,-2) node[draw,rectangle, minimum size=5mm] (F) {$\phi_1,\phi_2$};
        \draw (0,-2) node[draw,circle, minimum size=5mm] (G) {};
        \path[-latex'] (A) edge node[midway,above]{a} (B)
        (A) edge node[left] {b} (E)
        (E) edge node[above right] {a} (F)
        (E) edge node[left]{b}(G);
        \path[-latex'] (F) edge[loop right] (F)
        (G) edge[loop right] (G);
      \end{tikzpicture}
    \end{center}
    \caption{Example showing that $\AA\not\Rightarrow \AG^\land$. 
      Player~1 controls circle states and \player{2} square states;
      player~$3$ does not control any state.
    }
    \label{fig:aa-not-ag-and}
  \end{figure}
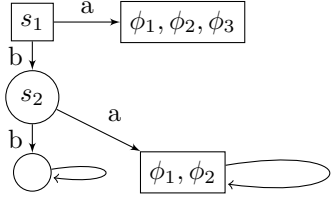

  \fbox{$\AA \not \Rightarrow \AG^\lor$}
  We will provide a counter-example to show our claim. Note that we need strictly more than two players
  since otherwise $\AG^\lor$ is equivalent to $\AG^\land$, and we have just shown that $\AA$ implies
  $\AG^\land$.
  
  Consider the game with three players in Fig.~\ref{fig:aa-not-dd}.
  Define the following objectives: $\phi_1 = \GF(s_4 \lor s_7)$,
  $\phi_2 = \GF(s_4 \lor s_6)$, $\phi_3 = \texttt{true}$,
  where~$\phi_i$ is player~$i$'s objective.
  These are actually reachability objectives since the game ends in absorbing states.

  Now, action~$b$ is dominated at states~$s_2$ and~$s_3$ for player~$2$. Thus player~$1$
  has a \AA-winning strategy which consists in taking~$a$ at~$s_1$.
  Player~$2$ has a winning strategy in the game (taking~$a$ at both states).
  Player~$3$ has a \AA-winning strategy too since actions~$b$ are eliminated for player~$2$.
  Therefore, there is an \AA-winning strategy profile which ends in~$s_4$.

  On the other hand, there is no $\AG^\lor$ profile. In fact, player~$1$ has no winning
  strategy to ensure $\phi_2 \lor \phi_3 \Rightarrow \phi_1$, which is equivalent
  to $\phi_1$ since~$\phi_3 = \texttt{true}$.
  \begin{figure}[ht]
    \begin{center}
      \begin{tikzpicture}[xscale=2]
        \draw (0,0) node[draw,rectangle,minimum size=5mm] (A) {$s_1$};
        \draw (1,0) node[draw,circle,minimum size=5mm] (B) {$s_3$};
        \draw (2,0) node[draw,circle, minimum size=5mm] (C) {$s_4$};
        \draw (2,-1) node[draw,circle, minimum size=5mm] (D) {$s_5$};
        \draw (0,-1) node[draw,circle,minimum size=5mm] (E) {$s_2$};
        \draw (0,-2) node[draw,circle, minimum size=5mm] (F) {$s_6$};
        \draw (1,-2) node[draw,circle, minimum size=5mm] (G) {$s_7$};
        \draw (2.8,0) node (lab1) {111};
        \draw (2.8,-1) node (lab2) {001};
        \draw (0,-2.6) node (lab3) {011};
        \draw (1,-2.6) node (lab3) {101};
        \path[-latex'] (A) edge node[midway,above]{a} (B)
        (B) edge node[above]{a} (C)
        (B) edge node[above]{b} (D)
        (A) edge node[left] {b} (E)
        (E) edge node[left] {a} (F)
        (E) edge node[above right]{b}(G);
        \path[-latex'] (F) edge[loop right] (F)
        (G) edge[loop right] (G)
        (C) edge[loop right] (C)
        (D) edge[loop right] (D);
      \end{tikzpicture}
    \end{center}
    \caption{Example showing that $\AA\not\Rightarrow \AG^\lor$. 
      Player~1 controls circle states and \player{2} square states;
      player~$3$ does not control any state.
      At each absorbing state, the given Boolean vector represents
      the set of players for which the state is winning.
    }
    \label{fig:aa-not-dd}
  \end{figure}
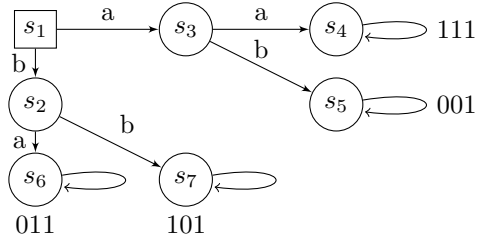

\end{proof}

\begin{lemma}
  For two player games, $\AA \Rightarrow \AG^\land$.
\end{lemma}
\begin{proof}
  Assume $\calG$ is a two player game and consider strategy profile~$(\sigma_1, \sigma_2)$ witness of \AA.
  Note that if \player{j} decreases its own value at position $k$ then its value for $h_{\le k+1}$ will be smaller or equal to $0$ which means \player{j} has no winning strategy from this history.
  By determinacy of turn-based zero-sum games, \player{}$3-j$ has a winning strategy for $\lnot \phi_j$.
  Therefore we can adjust the strategies $(\sigma_1,\sigma_2)$ such that if there is a player~$j$ that decreases its own value, the other player will make it lose.
  We write $(\sigma'_1,\sigma'_2)$ the strategies thus defined and we will show that they form a solution of Assume-Guarantee.

  Let $\rho$ be the outcome of the strategy profile $(\sigma'_1,\sigma'_2)$.
  We can show that $\rho$ is also the outcome of $(\sigma_1,\sigma_2)$.
  First we recall that an admissible strategy does not decrease its own value
  (Lemma~\ref{lemma:nodecrease}).
  Therefore each $\sigma'_i$ is identical to $\sigma_i$ on the run $\rho$.
  By Theorem~\ref{thm:aa}, $\rho$ satisfies $\phi_1\land \phi_2$. 

  Let~$\sigma_{1}''$ be an arbitrary strategy profile for~$1$, and consider $\rho' = \outcome_\Game(\sigma_1'',\sigma'_{2})$.
  We show that $\rho' \models \phi_1 \Rightarrow \phi_2$.
  Note that \player{2} cannot be the first to decrease its value during $\rho'$ since it behave according to $\sigma_2$ has long has there are no devition, and $\sigma_2$ is admissible and admissible strategies do not decrease their own values.
  \begin{itemize}
  \item
    If \player{1} decreases its value during $\rho'$, \player{2} will play to make him lose and $\rho' \not\models \phi_1$.
    As a consequence $\rho' \models \phi_1 \Rightarrow \phi_2$.
  \item
    Otherwise no player decreases its own value during $\rho'$.
    We assume that $\rho' \models \phi_1$ 
    and show that $\rho' \models \phi_2$.
    Since~$\rho' \models \phi_1$, by Lemma~\ref{lem:admissible-path}, there is a strategy $\tau''_1$ which is admissible and compatible with $\rho'$.
    Since $\rho'$ is an outcome of $\sigma_2'$, and of 
    $\tau''_1$,
    we have $\outcome_\Game(\tau''_1,\sigma_2) = \rho'$.
    Now, since $\tau''_{1}$ is 
    admissible and by the fact that $\sigma_2$ satisfies the condition 2 of \AA, we obtain $\rho' \models \phi_2$, which proves the property.
  \end{itemize}
  We can show the same property replacing the roles of \player{1} and \player{2}, thus showing that the profile is solution of $\AG^\land$.
\end{proof}


To prove that there is no implication for the edges that are not in the diagram of \figurename~\ref{fig:implications}, it remains to be shown that:
\begin{lemma}
  $\AA \not\Rightarrow \Win$, $\AG^\lor \not\Rightarrow \Win$, $\Coop \not\Rightarrow \AA$, $\Coop \not\Rightarrow \AG^\land$, $\Coop \not\Rightarrow \RS^\exists(\dom)$, and for all $\gamma \in \{\NE,\SPE,\dom\}$, $\RS^{\exists,\forall}(\gamma) \not\Rightarrow \Coop$.
\end{lemma}
\begin{proof}
  \fbox{$\AA \not\Rightarrow \Win$} 
  Towards a contradiction assume $\AA \Rightarrow \Win$, then since we have $\Win \Rightarrow \AG^\land$ (Lemma~\ref{lem:win->aa}), we would have $\AA \Rightarrow \AG^\land$ but this contradicts Lemma~\ref{lem:aa->ag}.
  
  \fbox{$\AG^\lor \not\Rightarrow \Win$}
  Towards a contradiction assume $\AG^\lor \Rightarrow \Win$, then since we have $\Win \Rightarrow \AA$ (Lemma~\ref{lem:win->aa}), we would have $\AG^\lor \Rightarrow \AA$ but this contradicts Lemma~\ref{lem:aa->ag}.
  
  \fbox{$\Coop \not\Rightarrow \AA$}
  In \figurename~\ref{fig:ag-not-aa}, we have an example of a game where there is no solution for $\AA$ (see the proof of Lemma~\ref{lem:aa->ag} for details), however there is a solution for $\Coop$: $(l,b)$.

  \fbox{$\Coop \not\Rightarrow \AG^\land$}
  Consider the example of \figurename~\ref{fig:coop-not-ag}.
  There is a solution for $\Coop$: \player{1} plays $a$.
  However there is no solution for $\AG^\land$: \player{2} has no strategy to ensure that $\phi_1 \implies \phi_2$.

  \fbox{$\Coop \not\Rightarrow \RS^\exists(\dom)$}
  Consider the example of \figurename~\ref{fig:coop-not-rs}.
  This example has a solution for $\Coop$, for instance $(l,ac)$ or $(r,bd)$.
  However \player{2} has no dominant strategy: $l$ looses against $bd$ so it is dominated by $r$, and $r$ looses against $ac$ so it is dominated by $l$.
  Therefore $\RS^\exists(\dom)$ has no solution.
  \begin{figure}[ht]
  \begin{minipage}{0.3\textwidth}
    \begin{center}
      \begin{tikzpicture}
        \draw (0,0) node[draw,rectangle,minimum size=8mm] (A) {1};
        \draw (2,0) node[draw,circle,minimum size=8mm] (B) {$\phi_1,\phi_2$};
        \draw (0.6,-1.6) node[draw,circle,minimum size=8mm] (C) {$\phi_1$};
        \path[-latex'] (A) edge node[midway,above]{a} (B)
        (A) edge node[right]{b} (C);
        \draw[-latex'] (-1,0) -- (A);
      \end{tikzpicture}
    \end{center}
    \caption{Example showing that $\Coop\not\Rightarrow \AG^\land$. 
      Player~1 controls the circle state.
    }
    \label{fig:coop-not-ag}
  \end{minipage}\hfill
    \begin{minipage}{0.42\textwidth}
  \begin{center}
  \begin{tikzpicture}[xscale=2]
    \draw (0,0) node[draw,circle,minimum size=8mm] (A) {$2$};
    \draw (1,1) node[draw,minimum size=8mm] (B) {$3$};
    \draw (1,-1) node[draw,minimum size=8mm] (BB) {$3$};
    \draw (2,1.5) node[draw,minimum size=8mm] (C) {$\phi_1,\phi_2,\phi_3$};
    \draw (2,0) node[draw,minimum size=8mm] (D) {$\varnothing$};
    \draw (2,-1.5) node[draw,minimum size=8mm] (E) {$\phi_1,\phi_2,\phi_3$};
    \draw[-latex'] (-0.5,0) -- (A);
    \draw[-latex'] (A) -- node[above] {$l$} (B);
    \draw[-latex'] (A) -- node[right] {$r$} (BB);
    \draw[-latex'] (B) -- node[above] {$a$} (C);
    \draw[-latex'] (B) -- node[below] {$b$} (D);
    \draw[-latex'] (BB) -- node[above] {$c$} (D);
    \draw[-latex'] (BB) -- node[above] {$d$} (E);

  \end{tikzpicture}
  \end{center}
  \caption{Example showing that $\Coop \not\Rightarrow \RS^\exists(\dom)$.
    Player~2 controls circle states, \player{3} square states and \player{1} does not control any state.}
  \label{fig:coop-not-rs}
    \end{minipage}
    \hfill
  \begin{minipage}{0.24\textwidth}
    \begin{center}
      \begin{tikzpicture}
        \draw (0,0) node[draw,circle,minimum size=8mm] (A) {2};
        \draw (0,-2) node[draw,circle,minimum size=8mm] (B) {$\phi_1$};
        \path[-latex'] (A) edge node[midway,right]{a} (B);
        \draw[-latex'] (0,1) -- (A);
      \end{tikzpicture}
    \end{center}
    \caption{Two-player game showing that $\RS^{\exists,\forall}(\gamma) \not\Rightarrow \Coop$.
      Player~2 controls the circle state but has no choice.
    }
    \label{fig:rs-not-coop}
  \end{minipage}
  \end{figure}

  \fbox{$\RS^{\exists,\forall}(\gamma) \not\Rightarrow \Coop$}
  Consider the example of \figurename~\ref{fig:rs-not-coop}.
  There is no solution for $\Coop$: \player{2} can never win.
  However there is a solution for any concept in $\RS^{\exists,\forall}(\gamma)$: \player{1} wins against any of the strategy satisfying these concepts since the only possible outcome is winning for him.
  
\end{proof}

This concludes the proof of Theorem~\ref{thm:implications}.

\subsection{Proof of Lemma~\ref{lem:aadv->aa}}
We now give the proof of Lemma~\ref{lem:aadv->aa}
\lmwuhaa*
\begin{proof}
  Assume that $\sigma_2^w$ is a winning strategy for $\phi_2$ and let $\sigma_1,\sigma_2$ be a solution of {\sf Win-under-Hyp}.
  We have that $\forall \sigma'_2.\ \sigma_1,\sigma'_2 \models \phi_2 \Rightarrow \phi_1$
  and $\forall \sigma'_1.\ \sigma'_1,\sigma_2 \models \phi_1 \Rightarrow \phi_2$.
  Since $\sigma_2^w$ is a winning strategy, all admissible strategies of player~2 are winning.
  Then, for all $\sigma'_2 \in \Adm_2$, we have $\Game,\sigma_1,\sigma'_2 \models \phi_2$ and because $\forall \sigma'_2, \sigma_1,\sigma'_2 \models \phi_2 \Rightarrow \phi_1$, we also have that $G, \sigma_1,\sigma'_2 \models \phi_1$.
  If $\sigma_1$ is dominated, there exists a non-dominated strategy $\sigma^a_1$ that dominates it~\cite[Thm.~11]{berwanger07}, otherwise we take $\sigma^a_1 = \sigma_1$.  
  In both cases $\sigma^a_1$ is admissible.
  As $\sigma_1$ is dominated by $\sigma^a_1$, $\Game , \sigma_1, \sigma'_2 \models \phi_1$ implies $\Game , \sigma^a_1, \sigma'_2 \models \phi_1$.
  This shows that the condition $\forall \sigma'_2 \in \Adm_2(\Game).\ \Game,\sigma_1^a,\sigma'_2\models \phi_1$ is satisfied.
  Since $\sigma^w_2$ is winning, it is admissible and we also have $\forall \sigma'_1 \in \Adm_1(\Game).\ \Game,\sigma'_1,\sigma^w_2\models \phi_2$.
  Therefore all conditions of Assume-Admissible are satisfied by $(\sigma^a_1,\sigma^w_2)$.
  
\end{proof}

We now prove Theorem~\ref{thm:rectangular}.
\thmrectangular*
\begin{proof}
  \fbox{\AA is rectangular}
  If there is no solution to \AA, then the set of witness is empty, and therefore is rectangular.
  If there is only one solution, then it is the Cartesian product of singletons and therefore also a rectangular set.

  Otherwise let $\sigma_\Agt$ and $\sigma'_\Agt$ be two solutions of \AA.
  Let $i$ be a player of $\Agt$, we show that $\sigma_i,\sigma'_{-i}$ is also a solution of \AA.
  We have that $\sigma_i \in \Adm(\Game)$ and for all $j\ne i$, $\sigma_j\in \Adm(\Game)$, because condition 1 holds for $\sigma_\Agt$ and $\sigma'_\Agt$.
  Therefore condition~1 holds for $\sigma_i,\sigma'_{-i}$.
  Similarly, $\forall \sigma'_{-i} \in \Adm_{-i}(\Game).\ \Game, \sigma'_i,\sigma_i \models \phi_i$ and for all $j\ne i$, $\forall \sigma'_{-j} \in \Adm_{-j}(\Game).\ \Game, \sigma'_j,\sigma_j \models \phi_j$, because condition 2 holds for $\sigma_\Agt$ and $\sigma'_\Agt$.
  Therefore condition~2 holds for $\sigma_i,\sigma'_{-i}$ and it is a witness of \AA.

  Let $\Sigma^{aa}_i$ be the set of strategy $\sigma_i$ such that there exists $\sigma_{-i}$ such that $\sigma_{i},\sigma_{-i}$ is a witness of \AA.
  We can show that the set of witness of \AA~is the Cartesian product of the $\Sigma^{aa}_i$.
  We obviously have that the set of solutions is included in $\prod_{i\in\Agt} \Sigma^{aa}_i$.
  Let $\sigma_\Agt$ be a profile in $\prod_{i\in\Agt} \Sigma^{aa}_i$, and $\sigma'_\Agt$ a witness of \AA.
  We can replace for one $i$ at a time, the strategy $\sigma'_i$ by $\sigma_i$ in $\sigma'_\Agt$ and by the small property we previously proved, the strategy profile stays a solution of \AA.
  Therefore $\sigma_\Agt$ is a solution of \AA.
  This shows that the set of solutions is the rectangular set $\prod_{i\in\Agt} \Sigma^{aa}_i$.

  \fbox{$\sigma'_C \in \adm_C(\Game) \Rightarrow \Game,\sigma_{-C},\sigma'_C \models \bigwedge_{i \in -C} \phi_i$}
  This claim follows from the definition of \AA-winning strategy profiles, since each strategy is winning against admissible strategies.

  \medskip

  Now
  Consider any game~$\Game$ and fix a profile~$\sigma_\Agt$ such that $\Game,\sigma_\Agt \models \bigwedge_{1\leq i \leq n} \phi_i$.

  \fbox{$\Win$}
  Assume $\sigma_\Agt$ is solution to $\Win$, then each~$\sigma_i$ is a winning strategy.
  Let $\sigma'_i$ be a strategy part of another profile solution to $\Win$.
  Then the strategy $\sigma'_i$ ensures~$\phi_i$ against any strategy profile for~$-i$.
  If we replace $\sigma_i$ by $\sigma'_i$ in the profile $\sigma_\Agt$ then the condition for $\Win$ are still satisfied.
  Thus the rule $\Win$ is rectangular.

  \fbox{$\RS^\forall(\dom)$}
  Let $\sigma_1$ be a solution of $\RS^\forall(\dom)$.
  Let $\sigma_{2},\dots,\sigma_n$ and $\sigma'_2,\dots,\sigma'_n$ be profiles of $\Sigma^\dom_{\Game,\sigma_1}$ such that  $\sigma_1,\sigma_{2},\dots,\sigma_n \models \phi_1$ and $\sigma_1,\sigma'_2,\dots,\sigma'_n \models \phi_1$.
  If we define a profile $\tau_2, \dots,\tau_n$ where each $\tau_i$ is either $\sigma_i$ or $\sigma'_i$, then as each $\tau_i$ is dominant, we have $\sigma_1,\tau_2,\dots,\tau_n \models \phi_1$ because $\sigma_1$ is a solution of $\RS^\forall(\dom)$.
  Therefore the profile belongs to $\Sigma^\dom_{\Game,\sigma_1}$ and makes $\sigma_1$ win.
  This shows that the rule is rectangular.

  \fbox{$\RS^\exists(\dom)$}
  Consider the example of \figurename~\ref{fig:rs-dom-not-rect}.
  Since \player{2} and \player{3} are always winning, all their strategies are dominant. 
  There is only one strategy~$\sigma_1$ for \player{1} since it controls no state.
  The profiles $(a,c)$ and $(b,d)$ are strategies of $\Sigma^\dom_{\Game,\sigma_1}$ such that $\sigma_1$ wins for $\phi_1$, but the profile $(\sigma_1,a,d)$ does not make~$\phi_1$ hold.
  The rule is therefore not rectangular.

\begin{figure}[ht]
  \begin{minipage}{0.45\textwidth}
    \centering
    \begin{tikzpicture}
      \draw (0,0) node[draw,minimum size=6mm] (A) {2};
      \draw (-1,-1) node[draw,circle,minimum size=6mm] (B) {3};
      \draw (1,-1) node[draw,minimum size=6mm] (C) {$\phi_{1},\phi_2,\phi_{3}$};
      \draw (-1.7,-2) node[draw,minimum size=6mm] (D) {$\phi_1,\phi_2,\phi_3$};
      \draw (-0.3,-2) node[draw,minimum size=6mm] (E) {$\phi_2,\phi_3$};
      
      \draw[-latex'] (A) -- node[above left] {a} (B);
      \draw[-latex'] (A) -- node[above right] {b} (C);
      \draw[-latex'] (B) -- node[above left] {c} (D);
      \draw[-latex'] (B) -- node[above right] {d} (E);
    \end{tikzpicture}
    \caption{Game with three players showing that rule $\textsf{RS}^\exists(\dom)$ is not rectangular.
    Here, player~$1$ controls no state;
    player~$2$ controls the square state, and player~$3$ controls the round state.
  }\label{fig:rs-dom-not-rect}
    \end{minipage}
    \hfill
    \begin{minipage}{0.45\textwidth}
    \centering{
    \begin{tikzpicture}
      \draw (0,0) node[draw,minimum size=6mm] (A) {2};
      \draw (-1,-1) node[draw,circle,minimum size=6mm] (B) {3};
      \draw (1,-1) node[draw,minimum size=6mm] (C) {$\phi_1,\phi_2,\phi_3$};
      \draw (-1.7,-2) node[draw,minimum size=6mm] (D) {$\phi_1,\phi_2,\phi_3$};
      \draw (-0.3,-2) node[draw,minimum size=6mm] (E) {$\phi_1,\phi_3$};
      
      \draw[-latex'] (A) -- node[above left] {a} (B);
      \draw[-latex'] (A) -- node[above right] {b} (C);
      \draw[-latex'] (B) -- node[above left] {c} (D);
      \draw[-latex'] (B) -- node[above right] {d} (E);
    \end{tikzpicture}}
    \caption{Game with three players showing that rule $\RS^\forall(\NE)$ and $\RS^\forall(\SPE)$ are not rectangular.
      Player~1 controls no state, \player{2} controls the square state and \player{3} the round state.}
    \label{fig:ex-rs-forall}
    \end{minipage}
\end{figure}

  \begin{figure}[htb]
    \begin{minipage}{0.45\textwidth}
    \centering{
    \begin{tikzpicture}
      \draw (0,0) node[draw,minimum size=6mm] (A) {2};
      \draw (-1,-1) node[draw,circle,minimum size=6mm] (B) {3};
      \draw (1,-1) node[draw,minimum size=6mm] (C) {$\phi_1,\phi_2,\phi_3$};
      \draw (-1.5,-2) node[draw,minimum size=6mm] (D) {$\phi_1,\phi_2,\phi_3$};
      \draw (-0.2,-2) node[draw,minimum size=6mm] (E) {$\phi_3$};
      
      \draw[-latex'] (A) -- node[above left] {a} (B);
      \draw[-latex'] (A) -- node[above right] {b} (C);
      \draw[-latex'] (B) -- node[above left] {c} (D);
      \draw[-latex'] (B) -- node[above right] {d} (E);
    \end{tikzpicture}}
    \caption{
      Game with three players showing that rule $\RS^\exists(\NE)$ and $\RS^\exists(\SPE)$ are not rectangular.
      Player~1 controls no state, \player{2} controls the square state and \player{3} the round state.}
    \label{fig:ex-spe}
    \end{minipage}
    \hfill
    \begin{minipage}{0.45\textwidth}
    \centering{
    \begin{tikzpicture}
      \draw (0,0) node[draw,circle,minimum size=6mm] (A) {1};
      \draw (-1,-1) node[draw,minimum size=6mm] (B) {2};
      \draw (1,-1) node[draw,minimum size=6mm] (C) {$\phi_1,\phi_2$};
      \draw (-1.5,-2) node[draw,minimum size=6mm] (D) {$\phi_1,\phi_2$};
      \draw (-0.5,-2) node[draw,minimum size=6mm] (E) {};
      
      \draw[-latex'] (A) -- node[above left] {a} (B);
      \draw[-latex'] (A) -- node[above right] {b} (C);
      \draw[-latex'] (B) -- node[above left] {c} (D);
      \draw[-latex'] (B) -- node[above right] {d} (E);
    \end{tikzpicture}}
    \caption{Game with two players showing that rule $\AG^\lor$ and $\AG^\land$ are not rectangular.
      Player~1 controls the round state and \player{2} the square state.}
    \label{fig:ex-ag}
    \end{minipage}
  \end{figure}

  \fbox{$\RS^\forall(\NE)$\ and $\RS^\forall(\SPE)$} 
  Consider the game represented in \figurename~\ref{fig:ex-rs-forall}.
  Player~1 has only one strategy~$\sigma_1$ and the other players have two possible strategies: $a$ and $b$ for \player{2} and $c$ and $d$ for \player{3}.
  Since \player{1} is always winning, $\sigma_1$ is a solution for $\RS^\forall(\NE,\SPE)$.
  The profiles $(a,c)$ and $(b,d)$ are two (subgame perfect) Nash equilibria which make $\phi_1$ hold. 
  However the profile $(a,d)$ obtained by picking one strategy in each profile, is no longer a Nash equilibrium (and so not a subgame perfect equilibrium).
  Therefore $\RS^\forall(\NE)$\ and $\RS^\forall(\SPE)$ are not rectangular.

  \fbox{$\RS^\exists(\NE)$\ and $\RS^\exists(\SPE)$} 
  Consider the game represented in \figurename~\ref{fig:ex-spe}.
  Player~1 has only one strategy and the other players have two possible strategies: $a$ and $b$ for \player{2} and $c$ and $d$ for \player{3}.
  The profiles $(a,c)$ and $(b,d)$ are two (subgame perfect) Nash equilibria which make $\phi_1$ hold. 
  However the profile $(a,d)$ obtained by taking one strategy in each profile, is no longer winning for \player{1}.
  Therefore $\RS^\exists(\NE)$\ and $\RS^\exists(\SPE)$ are not rectangular.

  \fbox{\Coop} 
  Once again, consider the game represented in \figurename~\ref{fig:ex-spe}.
  The profiles $(a,c)$ and $(b,d)$ make all players win, but the profile $(a,d)$, is no longer winning for the \player{1}, so it is not a solution of \Coop.
  Therefore \Coop is not rectangular.

  \fbox{$\AG^\lor$ and $\AG^\land$} 
  Consider the game represented in \figurename~\ref{fig:ex-ag}.
  The profiles $(a,c)$ and $(b,d)$ make the two players win.
  Since all possible outcome of the game satisfy the implications $\phi_1\Rightarrow \phi_2$ and $\phi_2\Rightarrow \phi_1$, both profiles  are solution to $\AG^\lor$ and $\AG^\land$ (note that the two concepts coincide here because there are only two players).
  However the profile $(a,d)$ obtained by taking one strategy in each profile, is no longer winning for the \player{1}.
  Therefore $\AG^\lor$ and $\AG^\land$ are not rectangular.
  
\end{proof}

\section{Complements on Algorithm for Assume-Admissible Synthesis  (Section~\ref{section:algo-aa})}
In this section, we recall some detailed results on values and
admissible strategies, and give the details of the algorithm for assume-admissible synthesis.

\subsection{Values of Histories and Admissibility}
For the proofs, we actually need a more refined notion of value, which coincides
with the notion defined in the core of our paper for prefix-independent objectives.
Let us give the formal definition from~\cite{berwanger07}.

We fix a game~$\Game$.
Given a history $h$, and a set of strategies $\Sigma_i'$ for player~$i$, 
we write $\Sigma'_i(h)$ for the set of strategies of $\Sigma'_i$ compatible with $h$, that is,
the set of strategies $\sigma_i$ such that $h\in \outcome_\Game(\sigma_i)$.
We also write $\Sigma'(h)$ for $\prod_{i\in\Agt} \Sigma'_i(h)$.

\begin{definition}[Value~\cite{berwanger07}]
	\label{def:value}
  Let $\Sigma'$ be a rectangular set of strategy profiles.
  The \newdef{value of history~$h$ for \player{i} with respect to $\Sigma'$}, written $\val_i(\Sigma',h)$, is given by:
  \begin{itemize}
  \item if every $\sigma_{\Agt} \in \Sigma'(h)$ 
    is losing for \player{i} then $\val_i(\Sigma',h)=-1$;
  \item if there is a strategy of $\sigma_i \in \Sigma'_i(h)$ such that for all strategy profiles $\sigma_{-i}$ in $\Sigma'_{-i}(h)$, 
    the profile $(\sigma_i,\sigma_{-i})$ is winning for \player{i} then $\val_i(\Sigma',h)=1$;
  \item otherwise $\val_i(\Sigma',h)=0$;
  \end{itemize}
\end{definition}

We use the shorthand notation $\val_i(h) = \val_i(\Sigma,h)$. 
Notice that the value only depends on the last state for prefix-independent objectives.
We may thus write $\val_i(s) = \val_i(h)$ for $s = \last(h)$; observe that this is the notation
we use in the core of the paper.

We noted in Lemma~\ref{lemma:nodecrease} that admissible strategies do not make the player's value decrease.
We prove here that conversely, any winning run on which player~$i$ does not decrease its own value is compatible with an admissible strategy of player~$i$.
\begin{lemma}\label{lem:admissible-path}
  Let $i$ be a player, and~$\rho$ a history.
  If $\rho \models \phi_i$ and \player{i} does not decrease its own value in any prefix of~$\rho$,
  then there exists a strategy profile $(\sigma_i,\sigma_{-i}) \in \Adm_i \times \Sigma_{-i}$ such that $\rho$ is the outcome of $(\sigma_i,\sigma_{-i})$.
\end{lemma}
\begin{proof}
  We define the strategies $\sigma_i$ and $\sigma_{-i}$ to follow $\rho$ when possible (if the current history is a prefix of $\rho$, then proceed to the following state of $\rho$), and if a deviation has occurred in $\rho$, that is there is $k$ such that $h_k = \rho_k$, $h_{k+1} \ne \rho_{k+1}$, then starting from $h_{\le k+1}$, $\sigma_i$ follows a non-dominated strategies with respect to 
$\Sigma(h_{k+1})$ (Thanks to \cite{berwanger07} and the fact that this set allows shifting~\cite{berwanger07} such a non-dominated strategies exists).
  The run $\rho$ is obviously an outcome of this profile. 
  We now have to show that the strategy $\sigma_i$ that we define is admissible.
  According to \cite[Lem.~9]{berwanger07}, it is enough to show that for every history $h$ outcome of $\sigma_i$, the value for \player{i} with respect to $\{\sigma_i\}\times \Sigma_{-i}$ is greater or equal to that of $\Sigma$.

  Let $h$ be a finite outcome of $\sigma_i$.
  We distinguish the case where $h$ has deviated from $\rho$ and the case where it has not.

  If a deviation has occurred, then $\sigma_i$ follows a strategy non dominated with respect to $\Sigma(h_{\le k+1})$ where $k$ is the last index where $h_k=\rho_k$.
  By \cite[Lem.~9]{berwanger07}, the value of $\{\sigma_i\}\times \Sigma_{-i}(h)$ in $h$ is greater or equal to that of $\Sigma(h_{\le k+1})$.
  Since $\Sigma_{-i}(h) \subseteq \Sigma_{-i}(h_{\le k+1})$, the value of $h$ with respect to $\{\sigma_i\}\times \Sigma_{-i}(h)$ is greater or equal to that with respect $\Sigma(h)$.
  Note that by the definition of the value, the value of $h$ with respect to a rectangular set $\Sigma'$ is equal to that of $h$ with respect to $\Sigma'(h)$.
  Therefore the value of $h$ with respect with $\{\sigma_i\} \times \Sigma$ is greater or equal to that with respect to $\Sigma$.

  If a deviation has not occurred then $h$ is a prefix of $\rho$.
  The value of $h$ with respect to $\Sigma$ is greater or equal to 0 since $\rho$ is winning for $\phi_i$.
  Then:
  \begin{itemize}
  \item if the value is 0, then as there is an outcome of $\sigma_i$ after this history which is winning (the run~$\rho$), the value of $\sigma_i$ is at least 0;
  \item if the value is 1, then we can show that from this state $\sigma_i$ plays a winning strategy: if we stay along $\rho$, the run is winning, if we deviate, the run goes to a state of value 1 (\player{i} has not decreased its own value and from a state of value 1 for \player{i}, any action of the adversaries lead to a state of value 1), and therefore $\sigma_i$ revert to a winning strategy.
  \end{itemize}
  Therefore the property is satisfied by $\sigma_i$ and it is admissible.  
  
\end{proof}

We now prove Lemma~\ref{lemma:adm-outcomes} characterizing the outcomes of admissible strategies.
This result follows from~\cite{BRS14}, but we adapt it to make the formulas~$\Phi_i$ appear explicitly.

\lmphii* 
\begin{proof}
  In~\cite[Lemma~6]{BRS14}, an automaton~$\calA_i^1$ is defined such that
  $\mathcal{A}_i^1 \cap \Out_\Game(\Sigma) = \Out_\Game(\Adm_i,\Sigma_{-i})$.
  (Note that a more general construction $\calA_i^n$ was given in~\cite{BRS14}; we only need the case $n=1$ here.)

  We now analyze further the language of $\mathcal{A}_i^1$. 
  The edges are those of $\G$ except for edges outside of $E_i$ (these edges are noted $T$ in \cite{BRS14}), so the set of runs in $\mathcal{A}_i^1$ corresponds to $\Out_\Game \cap \G(E_i)$.
  Now a run of $\mathcal{A}_i^1$ is accepted if, and only if one the following condition is satisfied, writing $\textsf{VR}(\rho)$ for the sequence $(\val(\rho_i))_{i\in \mathbb{N}}$:
  \begin{itemize}
  \item $\textsf{VR}(\rho) \in 0^* (-1)^\omega$;
  \item $\textsf{VR}(\rho) \in 0^* 1^\omega$ and $\rho \models \phi_i$;
  \item $\textsf{VR}(\rho) \in 0^\omega$ and $\rho \models \phi_i$ or $\rho \models \G \F(H_i)$.
  \end{itemize}
  
  Any run of $\Out_\Game \cap \G(E_i)$ reaching some state of value $-1$ is necessarily losing; thus all successors also have value $-1$.
  Similarly, because we removed edges where \player{i} decreases its own value, once the runs reaches a state of value $1$, it never gets out of these states.
  Therefore runs of $\Out_\Game \cap \G(E_i)$ have one of the three forms: $0^*(-1)^\omega$, $0^* 1^\omega$ or $0^\omega$.

  Let $\rho$ be a run that is accepted by $\mathcal{A}_i^1$, it satisfies $\G(E_i)$ and:
  \begin{itemize}
  \item if $\rho$ ends in the states of value $-1$ then it does not visit $V_{i,1}$ or $V_{i,0}$ infinitely often and thus belongs to $\Phi_i$;
  \item if $\rho$ ends in the states of value $1$, then by the acceptance condition it satisfies $\phi_i$
    and thus belongs to $\Phi_i$;
  \item otherwise it stays in the states of value $0$, then by the acceptance condition, either it satisfies $\phi_i$ or $\G\F(H_i)$ and thus belongs to $\Phi_i$.
  \end{itemize}

  Now let $\rho$ be a run that satisfies $\phi_i$, it satisfies $\G(E_i)$ and therefore corresponds to a valid run of $\mathcal{A}_i^1$.
  \begin{itemize}
  \item If $\rho$ ends in the states of value $-1$ then condition $\textsf{VR}(\rho) \in 0^* (-1)^\omega$ is satisfied, thus $\rho$ is accepted by $\mathcal{A}_i^1$.
  \item If $\rho$ ends in the states of value $1$, then by definition of $\Phi_i$ it satisfies $\phi_i$ and condition $\textsf{VR}(\rho) \in 0^* 1^\omega$ and $\rho \models \phi_i$ is satisfied, thus $\rho$ is accepted by $\mathcal{A}_i^1$.
  \item Otherwise it stays in the states of value $0$, then by definition of $\Phi_i$, either $\phi_i$ or $\G\F(H_i)$ holds for $\rho$, hence $\textsf{VR}(\rho) \in 0^\omega$ and $\rho \models \phi_i$ or $\rho \models \G \F(H_i)$ is satisfied, thus $\rho$ is accepted by $\mathcal{A}_i^1$.
  \end{itemize}


  This shows that $\Phi_1\cap \Out_\Game = \mathcal{A}^1_i \cap \Out_\Game$ and by \cite[Lemma~6]{BRS14}, this equals $\Out(\Adm_i,\Sigma_{-i})$.
  
\end{proof}

\subsection{Algorithm}
We start by a characterization of games with \AA-winning strategy profiles.
An assume-admissible winning strategy ensures its objective against all admissible strategies of other players;
we formalize this in the following lemma using values (see
		Definition~\ref{def:value}).

\begin{lemma}\label{lem:aa-val}
  For all game $\Game$,
  there is an \AA-winning strategy profile 
  if, and only if for all players~$i$, $\val_i((\Adm_j(\Game))_{j\in\Agt}, \sinit)= 1$.
\end{lemma}
\begin{proof}
  By definition of the value we have for all player~$i$, that:
  \begin{equation}\label{eq:val}
    \val_i((\Adm_i(\Game))_{i\in\Agt},\sinit) = 1 ~~ \Leftrightarrow  ~~ \exists \sigma_i \in \Adm_i(\Game).\ \forall \sigma_{-i} \in \Adm_{-i}(\Game).\ \Game,\sigma_i,\sigma_{-i} \models \win{i}
  \end{equation}
 
  \fbox{$\Rightarrow$}
  If $\sigma_\Agt$ is a witness for Assume-Admissible then for all players~$i$, $\sigma_i \in \Adm_i(\Game)$ is such that for all $\sigma_{-i} \in \Adm_{-i}(\Game)$, $\Game,\sigma_i,\sigma_{-i} \models \win{i}$.
  Hence by \eqref{eq:val}, for all players~$i$, $\val_i((\Adm_i(\Game))_{i\in\Agt},\sinit) = 1$.

  \fbox{$\Leftarrow$}
  Assume that for all players~$i$, $\val_i((\Adm_i(\Game))_{i\in\Agt},\sinit) = 1$.
  By \eqref{eq:val} we have that for all players~$i$, there exist $\sigma_i\in \Adm_i(\Game)$ such that $\forall \sigma_{-i} \in \Adm_{-i}(\Game).\ \Game,\sigma_i,\sigma_{-i} \models \win{i}$.
  Let $\sigma_\Agt$ be a strategy profile made of such strategies $\sigma_i$.
  We have that for all player $i$, $\sigma_i \in \Adm_i(\Game)$ (condition~1 in the definition of Assume-Admissible) and moreover for all $\sigma'_{-i} \in \Adm_{-i}(\Game)$, $\Game,\sigma_i,\sigma'_{-i} \models \win{i}$ (condition~2 in the definition of Assume-Admissible).
  Hence the strategy profile $\sigma_\Agt$ constitutes a solution to Assume-Admissible.
  
\end{proof}

The algorithm to synthesize \AA-winning strategies uses procedures from \cite{BRS14}, originally
developed to compute the \emph{outcomes} that survive the \emph{iterative elimination} of dominated strategies.
More precisely, the elimination procedure of~\cite{BRS14} first computes the outcomes of admissible strategies;
from this it deduces the strategies that are not dominated when all players are restricted to admissible strategies, and their possible outcomes;
and this is repeated until the set of outcomes stabilizes.
In the end, one obtains the set of the runs that are the outcomes of strategy profiles that have survived this iterative elimination.

Here, we roughly consider the first iteration of the above procedure, and explicitly give algorithms to actually synthesize
strategies that are winning against admissible strategies.

We explain the proof of Lemma~\ref{lemma:aa-omega} which follows from~\cite{BRS14}.
\lmomegai*
\begin{proof}
  It is shown in \cite[Prop.~5]{BRS14} that a strategy of \player{i} is a strategy of $\stratset^{n}_i$ which is winning from state $s$ against all strategies of $\stratset^{n}_{-i}$ if, and only if, it is winning for objective~$\Omega_i^n(s)$ (where $\stratset^{n}$ is the set of strategy that remain after $n$ step of elimination).
  The results immediately follows of the case $n=1$.
  
\end{proof}

We will now establish the correctness of the game~$\Game_i'$.
Let us first formalize the correspondence between $\Game$ and~$\Game_i'$. 
We define relation~$\mathord{\sim} \subseteq \Stat \times \Stat'$: for all $s \in \Stat \times \{\bot,0,\top\}$, $s \sim (s,x)$.
We extend this to runs by $\rho \sim \rho'$ iff for all $i \in \mathbb{N}$, $\rho_i \sim \rho'_i$.
In fact the relation defines a bijection:
\begin{lemma}
  For any $\rho \in \Out_\Game$ there is a unique $\rho' \in \out_{\Game_i'}$ such that $\rho \sim \rho'$.
\end{lemma}
\begin{proof}
  Assume towards a contradiction that we have $\rho'$ and $\rho''$ such that $\rho = \pi(\rho') = \pi(\rho'')$.
  Let $i$ be the last state such that they coincide: $\rho'_i = \rho''_i$ and $\rho'_{i+1} \ne \rho''_{i+1}$.
  Since $\pi(\rho') = \pi(\rho'')$ we have that they differ only by the second component.
  We can assume without loss of generality that there is are actions $a$ and $b$ such that $(\rho_i,a) \in E_j$ (where \player{j} controls $\rho_i$), $(\rho_i,b) \not\in E_j$ and $\delta(\rho_i,a)= \rho_{i+1} = \delta(\rho_i,b) $.
  This means that there is are actions $a$ and $b$ such that $(s,a) \in E_j$ (where \player{j} controls $\rho_i$), $(s,b) \not\in E_j$ and $\delta(\rho_i,a)= \rho_{i+1} = \delta(\rho_i,b) $.
  We have $\delta(\rho_i, b) = \delta(\rho_i,a)$, then by definition of $E_j$ and because $(\rho_i,a) \in E_j$, $\val_j(\delta(s,a)) = \val_j(s)$ therefore $\val_j(\delta(s,b)) = \val_j(s)$ and by definition of $E_j$, $(\rho_i,b)$ belongs to $E_j$ which contradicts our assumptions and ends the proof.
  
\end{proof}

We write $\pi$ the bijection which, to $\rho' \in \Out_{\Game_i'}$ associates $\rho \in \Out_{\Game}$ with $\rho \sim \rho'$.
We extend $\pi$ as a mapping from strategies of $\Game'_i$ to strategies of $\Game$ by $\pi(\sigma'_i) (h) = \sigma'_i(\pi^{-1}(h))$.
We can notice that for all strategies $\sigma'_i$, $\pi(\Out_{\Game_i'}(\sigma'_i)) = \Out_\Game(\pi(\sigma'_i))$.

\begin{lemma}\label{lem:game-prime-i}
  Let $\Game$ be a game, and $i$ a player. Player~$i$ has a winning strategy for $\Omega_i$ in~$\Game$ if, and only if, he has a winning strategy
  for $\Omega'_i$ in $\Game'_i$.
  Moreover if $\sigma'_i$ is winning for $\Omega'_i$ in $\Game'_i$ then $\pi(\sigma'_i)$ is winning for $\Omega_i$ in~$\Game$.
\end{lemma}
\begin{proof}
  We will first rewrite $\Omega_i$ in a form that is closer to that of $\Omega'_i$.
  The objective $\Omega_i$ is defined by $\Out_\Game(\Adm_i) \cap (\Out_\Game(\Adm_{-i}) \Rightarrow \phi_i)$.
  Observe that $\Out_\Game(\Adm_{-i}) = \cap_{j \neq i} \Out_\Game(\Adm_j)$ by definition.
  \begin{align*}
    \Omega_i & = \Out_\Game(\Adm_i) \cap \left(\bigcap_{j\ne i} \Out_\Game(\Adm_{j}) \Rightarrow \phi_i\right) \\
    \Omega_i & = \Phi_i \cap \Out_\Game \cap \left(\left(\Out_\Game \cap \bigcap_{j\ne i} \Phi_{j}\right) \Rightarrow \phi_i\right) ~ \text{using Lem.~\ref{lemma:adm-outcomes}} \\
    \Omega_i & = \Phi_i \cap \Out_\Game \cap \left(\bigcap_{j\ne i} \Phi_{j}  \Rightarrow \phi_i\right) \\
    \Omega_i & = \Out_\Game \land \G(E_i) \land M_i \land \left( \left(\bigcap_{j\ne i} M_{j}  \Rightarrow \phi_i\right) \lor \bigvee_{j\ne i} \lnot \G(E_j)\right)\\
    \Omega_i & = \Out_\Game \land \G(E_i) \land M_i \land \left( \left(\bigcap_{j\ne i} M_{j}  \Rightarrow \phi_i\right) \lor \bigvee_{j\ne i} \F(\lnot E_j)\right)
  \end{align*}

  \fbox{$\Rightarrow$}
  Let $\sigma_i$ be a winning strategy for $\Omega_i$ in $\Game$.
  We consider the strategy $\sigma'_i$ defined by $\sigma'_i(h') = \sigma_i(\pi(h'))$ and will show that it is winning for $\Omega'_i$.
  Let $\rho'$ be an outcome of $\sigma'_i$. 
  We have that $\pi(\rho')$ is an outcome of $\sigma_i$.
  Since $\sigma_i$ is winning for $\Omega_i$, $\pi(\rho')$ belongs to $\Omega_i$.

  \begin{itemize}
  \item If $\pi(\rho') \models M_i \land \G(E_i) \land \bigvee_{j\ne i} \F(\lnot E_j)$, then by construction of $\delta'$ the play $\rho'$ reaches a state of $\Stat \times \{\top\}$ and, from there, only states of $\Stat \times \{\top\}$ are visited.
    The condition $\GF(\Stat \times \{ \top\}) \land M_i$ is met by $\rho'$ and therefore $\rho'$ is winning for~$\Omega'_i$.
  \item Otherwise $\pi(\rho') \models M_i \land \G(E_i) \land (\bigwedge_{j\ne i} M_j) \Rightarrow \phi_i$.
    By construction of $\delta'$ the play $\rho'$ stays in $\Stat\times \{0\}$.
    The condition $\GF(\Stat \times \{ 0\}) \land M_i\land (\bigwedge_{j\ne i} M_j) \Rightarrow \phi_i$ is met by $\rho'$ and therefore $\rho' \in \Omega'_i$.
  \end{itemize}
  This shows that the strategy $\sigma'_i$ is winning for $\Omega'_i$ in $\Game'_i$.

  \fbox{$\Leftarrow$}
  Let $\sigma'_i$ be a winning strategy for $\Omega'_i$ in $\Game'_i$, we show that $\pi(\sigma'_i)$ is winning for $\Omega_i$ in $\Game$.
  Let $\rho$ be an outcome of $\pi(\sigma'_i)$.
  We have that $\pi^{-1}(\rho)$ is an outcome of $\sigma'_i$.
  Since $\sigma'_i$ is winning for $\Omega'_i$, $\pi^{-1}(\rho)$ belongs to $\Omega'_i$.
  We have that $\pi^{-1}(\rho) \models \GF(\Stat\times \{ 0,\top\})$ and by construction of $\delta'$ this ensures that all edges that are taken belong to $E_i$ and thus $\pi^{-1}(\rho)$ satisfies the condition $\G(E_i)$.
  \begin{itemize}
  \item If $\pi^{-1}(\rho) \models \GF(S \times \{\top\}) \land M_i$ then by construction of $\delta'$, an edge outside of $E_j$ for some $j\ne i$ is taken.
    This ensures condition $\F(\lnot E_j)$ and therefore $\rho$ belongs to $\Omega_i$.
  \item otherwise $\pi^{-1}(\rho) \models (\bigwedge_{j\ne i} M_j \Rightarrow \phi_i)$ and therefore $\rho$ satisfies the condition $\G(E_i) \land M_i \land \left(\bigcap_{j\ne i} M_{j}  \Rightarrow \phi_i\right)$ and hence belongs to $\Omega_i$.
  \end{itemize}
  This shows that the strategy $\pi(\sigma'_i)$ is winning for $\Omega_i$ in $\Game$.
  
\end{proof}

This characterization immediately yields the decidability of the problem in \PSPACE{}.
Moreover, \PSPACE-hardness follows from that of Muller games.
We also provide a polynomial-time algorithms for the particular case of B\"uchi
conditions.
The following theorem proves the first statement of Theorem~\ref{thm:aaalgo};
the second statement is proved afterwards (Theorem~\ref{thm:aa-compute}).
\begin{theorem}
  \AA-synthesis in multiplayer Muller games is \PSPACE-complete and \P-complete for B\"uchi objectives.
\end{theorem}
\begin{proof}
  We use the result of \cite[Prop.~8]{BRS14} in the special case where $n=1$,
  in which case the proposition says that for multi-player games with Muller objectives, checking whether $\val_i((\Adm_i(\Game))_{i\in\Agt},\sinit) = 1$ is in \PSPACE.
  By Lem.~\ref{lem:aa-val}, with one call to this algorithm for each player, we are able to decide if \AA has a solution and this proves \PSPACE~membership.
  For hardness we encode 2-players zero-sum Muller games in our setting: the first player keep the same objective and the second one is always winning.
  In this setting no strategy of the second player is dominated so finding an \AA winning strategy for the first player is the same as finding a winning strategy in the original game.

  \paragraph{B\"uchi Case}
  In the case of B\"uchi objectives, let us write $\phi_i = \G \F (B_i)$ and consider the objective $M_i = \GF(V_{i,1}) \Rightarrow \G\F(B_i) \land \G\F(V_{i,0}) \Rightarrow (\G\F(B_i) \lor \G\F(H_i))$.
  In $\Game$, a run that verifies $\G(E_i)$ will either visit only $V_{i,1}$ after some point, or only $V_{i,-1}$ after some point, or only $V_{i,0}$ (see the proof of Lemma~\ref{lemma:adm-outcomes} for details).
  This means that $M_i$ coincide with $\G\F((V_{i,1} \land B_i) \lor (V_{i,0} \land B_i) \lor (V_{i,0} \land H_i) \lor V_{i,-1})$ on the language $\Out_\Game(\Sigma) \cap \G(E_i)$.
  This is a B\"uchi condition.
  Note that in $\Game'_i$, runs that take an edge outside of $E_i$ end in the $\Stat\times\{\bot\}$ component and $\Omega'_i$ disallows such runs.
  Therefore, in the expression of $\Omega'_i$, we can also replace $M'_i$ by a B\"uchi condition that we write $\G\F(B^M_i)$.
  By abuse of notation we will also write $B_i$ for the states $B_i \times \{\bot,0,\top\}$ of the game $\Game'_i$.
  Consider now the objective $\Omega'_i$, it is given by $(\G\F(\Stat \times \{0\}) \land \G\F(B^M_i) \land (\bigwedge_{j\ne i} \G\F(B^M_j) \Rightarrow \G\F(B_i \times \{\bot,0,\top\}))) \lor (\G\F(\Stat \times \{\top\}) \land \G\F(B^M_i))$.
  Since in this game, states of $\Stat\times \top$ and $\Stat\times\bot$ are absorbing (no play can get out of those components) we write an equivalent objective which is:
  $(\G\F(B^M_i\times \{0\}) \land (\bigwedge_{j\ne i} \G\F(B^M_j) \Rightarrow \G\F(B_i))) \lor (\G\F(B^M_i \times \{\top\})$.
  We define a (small) deterministic parity automaton $\mathcal{A}$ that recognizes this language.
  Its state space is $(\{s,t,u,v\} \times (\{ j \mid j \in \Agt \setminus\{ i \}\} \cup \{ \top \}))$, the transition relation is a product of transitions for the two components:
  $s\xrightarrow{B_i^M \times \{0,\top\}} u$, $u \xrightarrow{\lnot B_i^M} t$, $t,u \xrightarrow{B_i^M \setminus B_i}$, $t,u \xrightarrow{B_i} v$, $v \xrightarrow{\tt} s$,
  and $j\xrightarrow{\lnot B_j \times \{0\}} j$, $j \xrightarrow{B_j \times \{0\}} j'$ where $j'$ is $j+1$ if $j+1 \in \Agt\setminus \{i\}$, $j+2$ if $j+1=i$ and $j+2\in \Agt$, $\top$ otherwise, $\top \xrightarrow{\tt} j_0$ where $j_0$ is the smallest element of $\Agt \setminus \{i\}$.
  The coloring is defined by a function $\chi$ where $\chi(v,\ast) = 4$ (where $\ast$ is any possible second component), $\chi(\{s,t,u\},\top)=3)$, $\chi(u,\Agt\setminus\{i\}) = 2$, and for all other states~$s$, $\chi(s)=1$.
  A word is accepted by $\mathcal{A}$ when the maximal color appearing infinitely often is even.

  We show that a play of $\Game'_i$ satisfies $\Omega'_i$ if, and only if, it is a word accepted by~$\mathcal{A}$.

  Let $\rho$ be a play of $\Game'_i$ which satisfies $\Omega'_i$, either it ends in the $\Stat \times \top$ component or the $\Stat \times 0$ component:
  \begin{itemize}
  \item If $\rho$ ends in the $\top$ component then the state of color~$3$ will not be visited infinitely often (we need to be in $\top$ states to progress on this component of the automaton).
    As $\rho$ visits infinitely often $B_i^M$, the corresponding run in $\mathcal{A}$ will visit infinitely often $u$, and therefore the maximal color that appears infinitely often is at least $2$; and since it is not $3$, it has to be even.
  \item Otherwise $\rho$ ends in the $0$ component.
    Since $\rho$ satisfies $\Omega'_i$, it visits $B_i^M$ infinitely often and either there is a $B_j^M$ for $j\ne i$ that is not visited infinitely often, or $\rho$ visits infinitely often $B_i$.
    \begin{itemize}
    \item If there is a $B_j^M$ for $j\ne i$ that is not visited infinitely often, then 
      the second component of $\mathcal{A}$ will get stuck at some point and its state $\top$ will not be visited infinitely often the state of color $3$.
      As $\rho$ visits infinitely often $B_i^M$, the corresponding run in $\mathcal{A}$ will visit infinitely often $u$, and therefore the maximal color that appears infinitely often is at least $2$; and since it is not $3$, it has to be even.
    \item Otherwise $\rho$ visits infinitely often $B_i$.
      Since we also visit $B_i^M$ infinitely often, the run of $\mathcal{A}$ corresponding to $\rho$ will reach infinitely often a state $(v,\ast)$ and therefore the maximal color occurring infinitely often is $4$.
    \end{itemize}
  \end{itemize}
  This proves that the word is accepted by $\mathcal{A}$.

  \medskip

  Now let $\rho$ be a play of $\Game'_i$ such that the corresponding word is accepted by~$\mathcal{A}$.
  If it is accepted then either the color $4$ is seen infinitely often or the color $2$ is and the color $3$ is not:
  \begin{itemize}
  \item If the color $4$ is visited infinitely often then this means $t$ is reached infinitely often, and because of the structure of $\mathcal{A}$, $u$ also is, which means both $B_i^M\times \{0\}$ and $B_i$ occur infinitely often.
    This implies that the run $\rho$ belongs to $\Omega'_i$.
  \item Otherwise the color $2$ is visited infinitely often and $3$ is not.
    The states $(\ast,\top)$ are therefore not visited infinitely often (otherwise the maximal color would be $3$ or $4$.
    We deduce from that and the structure of $\mathcal{A}$ that some $B_j^M$ for $j\ne i$ is not visited infinitely often.
    This means $\bigwedge_{j\ne i} \G\F(B_j^M)$ is not true for $\rho$.
    Since the color $2$ is seen infinitely often, this means $u,\ast$ is seen infinitely often and therefore $B_i \times \{ 0, \top\}$.
    This ensures $\rho$ belongs to $\Omega'_i$.
  \end{itemize}
  This proves that a play of $\Game'_i$ satisfy $\Omega'_i$ if, and only if, it is a word accepted by~$\mathcal{A}$.

  Then solving the game $\Game'_i$ with objective $\Omega'_i$ is the same as solving it with objective given by $\mathcal{A}$.
  This can be done by solving the parity game obtained by the product of $\Game'_i$ with the automaton $\mathcal{A}$.
  The obtained game is of polynomial size and the number of priority is $4$, such games can be solved in polynomial time~(see for instance \cite{LBCJM94,Seidl96}) and therefore we can decide our problem in polynomial time.
  
\end{proof}

We are now interested in \emph{computing} an \AA-winning strategy profile.
Thanks to Lemma~\ref{lem:game-prime-i}, we obtain an algorithm to compute \AA-winning strategies by looking for winning strategies in $\Game'_i$ and projecting them:

\begin{theorem}
  \label{thm:aa-compute}
  Given a game~$\Game$ with Muller objectives, if \AA has a solution, then an \AA-winning strategy profile can be computed in exponential time.
\end{theorem}
\begin{proof}
  If \AA has a solution, then by Lemma~\ref{lem:game-prime-i}, there is a winning strategy for $\Omega'_i$ in $\Game'_i$.
  This Muller game has polynomial size, hence we can compute a winning strategy $\sigma'_i$ in exponential time~(for instance in \cite{NRZ14} the authors show that we can compute such a winning strategy via a safety game of size $|\Stat|!^3$).
  By Lemma~\ref{lem:game-prime-i}, the projection $\pi(\sigma'_i)$ is an \AA-winning strategy.
  Doing this for each player we obtain a strategy profile solution of \AA.
  
\end{proof}

\section{Complements on Abstraction (Section~\ref{section:abstraction})}
\subsection{Abstract Games}
The pair of abstraction and concretization functions $(\alpha,\gamma)$
actually defines a \emph{Galois connection}:
\begin{lemma}
  The pair $(\alpha,\gamma)$ is a \emph{Galois connection}, that is, for all $S \subseteq \Stat$ and $T \subseteq \Stat^\abs$, we have that $\alpha(S) \subseteq T$ if, and only if, $S\subseteq \gamma(T)$.
\end{lemma}
\begin{proof}
  \fbox{$\Rightarrow$}
  Let $s\in S$.
  Since $\gamma$ defines a partition of $\Stat$, there exists $t\in \Stat^\abs$ such that $s \in \gamma(t)$.
  By definition of $\alpha$, $\alpha(s) = t$.
  Assuming $\alpha(S) \subseteq T$, we have that $t \in T$.
  As $s\in \gamma(t)$, we have $s \in \gamma(T)$.

  \fbox{$\Leftarrow$}
  If $s^\abs \in \alpha(S)$, then there is $s\in S$ such that $s^\abs = \alpha(s)$.
  Assuming $S \subseteq \gamma(T)$, there is $t\in T$ such that $s = \gamma(t)$.
  By definition of $\alpha$, we have that $\alpha(s) = t$.
  Therefore $s^\abs \in T$.
  
\end{proof}

We prove the soundness of the abstract arenas $\AbsA^C$ we defined, by showing that
if coalition~$C$ achieves an objective in~$\Game$, then it also achieves the objective
in~$\AbsA^C$.

\lmwinabs*
\begin{proof}
  Assume $\sigma_C$ is a winning profile of coalition $C$, for objective $\phi_k$ in $\Game$.
  We define by induction a winning strategy $\sigma^\abs_C$ in $\Game^{\abs,k,C}$.
  We assume that $\sigma^\abs_C$ has been defined in a manner such that for each finite outcome $h^\abs$ of $\sigma^\abs_C$ shorter than some bound $m$, there is some $h\in \gamma(h^\abs)$ such that $h$ is a finite outcome of $\sigma_C$.
  The idea is then to define $\sigma^\abs_C$ to resolve the determinism in a way which simulates the behavior from $h$.
  \begin{itemize}  
  \item If $s_i^\abs \in \bigcup_{i\in\Agt} \Stat^\abs_i \times \Act_i$, then $\sigma^\abs_C(h^\abs \cdot (\last(h^\abs),a)) = \gamma(t)$ where $t = \delta(\last(h),a)$.
  \item If $s^\abs \in \bigcup_{i\in C} \Stat^\abs_i$, $\sigma^\abs_C(h^\abs \cdot (\last(h^\abs),a) \cdot  s^\abs) = \sigma_C(h \cdot \delta(\last(h),a))$.
  \end{itemize}

  With this definition, our induction hypothesis will be respected for histories containing one more step, and therefore this holds for all histories.
  Let now $\rho^\abs$ be an outcome $\sigma^\abs_C$.
  By the way we defined this strategy there is a run $\rho$ outcome of $\sigma_C$ such that $\rho\in \gamma(\rho^\abs)$.
  As $\sigma_C$ is winning, $\rho$ satisfies the Muller condition $\phi_k$ and since $\gamma$ is compatible with players' objectives, $\rho^\abs$ satisfies $\phi_k^\abs$.
  Which show that $C$ has a winning strategy in $\Game^{\abs,k,C}$ for $\phi_k^\abs$.
  
\end{proof}

\subsection{Abstract Assume-Admissible Synthesis}
\paragraph{Value-Preserving Strategies}
We provide the proofs of the lemmas stated in the core of the paper.
\lmapprv*
\begin{proof}
  \fbox{$\overline{V}_{k,1}$} This is a direct consequence of Lemma~\ref{lem:winning-abstract}.

  \fbox{$\overline{V}_{k,-1}$} If $s\in V_{k,-1}$ then the coalition $\Agt$ has no winning strategy in $\Game$.
  By determinacy, the empty coalition has a strategy to ensure $\lnot \phi_k$.
  Therefore by Lemma~\ref{lem:winning-abstract}, the coalition $\varnothing$ has a strategy in
  $\AbsA^{\Agt}$ from $\alpha(s)$ that ensures  $\lnot \phi_k$.
  Therefore $s \in \gamma( \overline{V}_{k,-1})$.

  \fbox{$\overline{V}_{k,0}$}
  Recall that $V_{k,0} = \Win_{\Agt \setminus \{k\}}(\A,\lnot\phi_k) \cap \Win_{\Agt}(\A,\phi_k)$.
  Let $s$ be a state in $V_{k,0}$. By Lemma~\ref{lem:winning-abstract}, $\alpha(s)$ belongs
  to both sides of the intersection, thus $\alpha(s) \in \overline{V}_{k,0}$.
  Thus $V_{k,0} \subseteq \gamma(\overline{V}_{k,0})$.
  
  \fbox{$\underline{V}_{k,1}$}
  If $s^\abs\in \underline{V}_{k,1}$ then the coalition $\Agt\setminus\{k\}$ has no strategy in $\AbsA^{\Agt\setminus\{k\}}$ for $\lnot\phi_k^\abs$.
  Therefore by Lemma~\ref{lem:winning-abstract}, it has no strategy in $\A$ from any state of $\gamma(s^\abs)$ to do so.
  Therefore $k$ has a winning strategy in $\A$ from~$\gamma(s^\abs)$, and $\gamma(s^\abs) \in V_{k,1}$.
 
  \fbox{$\underline{V}_{k,-1}$} If $s^\abs \in \underline{V}_{k,-1}$, then the coalition $\Agt$ has no winning strategy in $\AbsA^{\Agt}$ for objective
  $\phi^{\abs}_k$.
  Therefore by Lemma~\ref{lem:winning-abstract}, it has no winning strategy in $\A$ from $\gamma(s^\abs)$ neither
  for the objective $\phi_k$.
  This means that $\gamma(s^\abs) \in V_{k,-1}$.

  \fbox{$\underline{V}_{k,0}$}
  Note that by definition of the $\nu X.$ operator, $\underline{V}_{k,0} \subseteq F$. Thus, let us just show that $\gamma(F) \subseteq V_{k,0}$.
  Recall that $V_{k,0} = \Win_{\Agt\setminus \{k\}}(\A, \lnot \phi_k) \cap \Win_{\Agt}(\A,\phi_k)$.  
  Let $s \in \gamma(\underline{V}_{k,0})$. Then player~$k$ has no strategy in $\AbsA^{\{k\}}$ for~$\phi_k^\abs$,
  hence, by Lemma~\ref{lem:winning-abstract}, it cannot win $\A$ neither for $\phi_k$ from~$\gamma(s)$.
  This shows that $\gamma(s) \subseteq \Win_{\Agt\setminus \{k\}}(\A, \lnot \phi^{\abs}_k) $.
  Furthermore, the coalition $\emptyset$ has no strategy in $\AbsA^\Agt$ for $\lnot \phi^{\abs}_k$, thus
  it does not have one neither in $\A$ for $\lnot \phi_k$ from $\gamma(s)$. In other terms,
  $\gamma(s) \subseteq \Win_{\Agt}(\A,\phi_k)$.
    
\end{proof}

We show that when playing according to~$\underline{E}_k$, player~$k$ ensures staying in~$\underline{V}$.
This is proven in the following.
Let us write $\gamma(\calE) = \{ (s,a) \mid (\alpha(s),a) \in \calE\}$ for $\calE \in \{ \underline{E}_k, \overline{E}_k\}$.
\begin{restatable}{lemma}{lemek}
\label{lem:Ek}
  For all games~$\Game$, and players~$k$,
  \begin{inparaenum}
  \item $\gamma(\underline{E}_k \cap (\underline{V} \times \Act)) \subseteq E_k \subseteq \gamma(\overline{E}_k)$.
  \item For all $s^\abs \in \S^\abs_k$, there exist $a,a'\in \Act_k$ such that $(s^\abs,a) \in \underline{E}_k$
    and $(s^\abs,a') \in \overline{E}_k$.
  \item 
    For all $(s^\abs,a) \in \underline{E}_k$ with $s^\abs \in \underline{V}$,
    we have $\post_\Delta(s^\abs,a) \subseteq \underline{V}$.
  \end{inparaenum}
\end{restatable}
\begin{proof}
  The inclusion $E_k \subseteq \gamma(\overline{E}_k)$ follows from the definition of $\overline{E}_k$, and by Lemma~\ref{lemma:approximate-v}.
  It also follows that for all $s \in \S^\abs_k$, there is $(s,a') \in \overline{E}_k$, since this is always the case for~$E_k$.

  Let $(s^\abs,a)$ be an edge in $\underline{E}_k \cap (\underline{V} \times \Act)$.
  Let $s$ be a state in $\gamma(s^\abs)$. 
  We have that $s \in \gamma(\underline{V}_{k,x})$ for some $x \in \{ -1,0,1\}$ and by Lemma~\ref{lemma:approximate-v} $s \in V_{k,x}$.
  By definition of $\underline{E}_k$, for all $t^\abs$ such that $\Delta(s^\abs,a,t^\abs)$, $t^\abs \in \underline{V}_{k,l}$ with $l \ge x$ and $s^\abs\in \underline{V}_{k,x}$.
  By Lemma~\ref{lemma:approximate-v}, we have that the value of all states in $\gamma(t^\abs)$ are at least as great as any state in $\gamma(s^\abs)$. 
  By definition of $\Delta$, $\alpha(\delta(s,a)) \subseteq \{ t^\abs \mid \Delta(s^\abs,a,t^\abs) \}$.
  Therefore $\alpha(\delta(s,a)) \in \cup_{l \geq x} \underline{V}_{k,l}$, which means $\delta(s,a) \in \cup_{l \geq x} \gamma(\cup_{l \geq x} \underline{V}_{k,l}) \subseteq \cup_{l\geq x} V_{k,l}$ using Lemma~\ref{lemma:approximate-v}.
  By definition of $E_k$ this implies that $(s,a) \in E_k$.

  It remains to prove that for all $s^\abs \in \S_k^\abs$, there is $(s^\abs,a) \in
	\underline{E}_k$, and that if $s^\abs \in \underline{V}$, then for all $(s^\abs,a) \in \underline{E}_k$,
  $\Delta(s^\abs,a,{t}^\abs)$ implies ${t}^\abs \in \underline{V}$.

  If $s^\abs \in \Stat_k^\abs \setminus \underline{V}$, then $(s^\abs,a) \in \underline{E}_k$ for all $a \in \Act_k$ 
  by definition.
  Let us now assume $s^\abs \in \underline{V}$. 
  \begin{itemize}
  \item 
    If $s^\abs \in \underline{V}_{k,-1}$, 
    then By definition of~$\underline{V}_{k,-1}$, we have that for all actions~$a$, and all states~$t^\abs$, if $\Delta^\abs(s^\abs,a,t^\abs)$ then
    $t^\abs \in \underline{V}_{k,-1}$. Thus $(s^\abs,a) \in \underline{E}_k$, and $t^\abs \in \underline{V}_{k,-1}$
    for any such~$t^\abs$, so $t^\abs \in \underline{V}$.

  \item 
    If $s \in \underline{V}_{k,1}$, then 
    there exists~$a$ such that $(s^\abs,a,t^\abs) \in \Delta^\abs$ implies
    $t^\abs \in \underline{V}_{k,1}$. So $(s^\abs,a) \in \underline{E}_k$,
    and $t^\abs \in \underline{V}_{k,1}$. Moreover this holds
    for all~$a$ with $(s^\abs,a) \in \underline{E}_k$, since for such~$a$, $(s^\abs,a,t^\abs) \in
    \Delta^\abs$ implies $t^\abs \in \underline{V}_{k,1}$ by definition of~$\underline{E}_k$.

  \item 
    If $s \in \underline{V}_{k,0}$, then 
    by the greatest fixpoint defining $\underline{V}_{k,0}$, there exists~$a \in \Act_k$
    such that for all $t^\abs$ with $\Delta(s^\abs,a,t^\abs)$, $t^\abs \in \underline{V}_{k,0}$.
    Conversely, for all $(s^\abs,a) \in \underline{E}_k$, $a$ ensures staying
    inside $\underline{V}_{k,0} \cup \underline{V}_{k,1}$.
    Thus for any such~$a$, $(s^\abs,a) \in \underline{E}_k$, and any~$t^\abs$, $\Delta(s^\abs,a,t^\abs)$ means
    $t^\abs \in \underline{V}_{k,0}$.
  \end{itemize}

\end{proof}

Recall that~$\underline{E}_k$ does not constrain the actions outside the set~$\underline{V}$;
thus strategies in $\strat_k(\underline{E}_k)$ can actually choose dominated actions outside~$\underline{V}$.
To prove that $\strat_k(\underline{E}_k)$ is an under-approximation of~$\strat_k(E_k)$ when started in~$\underline{V}$,
we need to formalize the fact that admissible strategies may choose arbitrary actions at states
that are not reachable by any outcome. Intuitively, such strategies cannot be dominated since
the dominated behavior is never observed.

For any strategy~$\sigma$, let $\reach(\Game,\sigma)$ denote the set of states reachable from~$\sinit$ by runs compatible with~$\sigma$.
We show that if one arbitrarily modifies an admissible strategy outside the set $\reach(\Game,\sigma)$, the resulting strategy is still admissible.
\begin{lemma}
  \label{lemma:admissible-non-reachable}
  Let~$\sigma$ be a strategy in $\adm_i(\Game)$ and $\sigma'$ a strategy in $\Sigma_i(\Game)$.
  If for all histories~$h$ such that $\last(h) \in \reach(\Game,\sigma)$, $ \Rightarrow \sigma(h) = \sigma'(h)$, then $\sigma' \in \adm_i(\Game)$.  
\end{lemma}
\begin{proof}
  For all profiles $\sigma_{-k} \in \Sigma_{-k}(\Game)$, we have $\outcome_\Game(\sigma_{-k},\sigma) = \outcome_\Game(\sigma_{-k},\sigma')$
  so if $\sigma'$ is dominated, then $\sigma$ would also be dominated, which is a contradiction.
  
\end{proof}

\lmadmstrat*
\begin{proof}
  Since $E_k \subseteq \gamma(\overline{E}_k)$ by Lemma~\ref{lem:Ek},
  we have  $\strat_k(E_k) \subseteq \gamma(\strat_k(\gamma(\overline{E}_k)))$.

  Assume $\sinit \in \gamma(\underline{V})$. 
  The fact that $\strat_k(\underline{E}_k)$, thus also $\gamma(\strat_k(\underline{E}_k))$ are non-empty
  follows from Lemma~\ref{lem:Ek} too, since for any state $s^\abs$  there is $a \in \Act_k$ with
  $(s^\abs,a) \in \underline{E}_k$.

  We prove that $\reach(\AbsA^{\Agt\setminus\{k\}},\sigma) \subseteq \underline{V}$ for all $\sigma \in \strat_k(\underline{E}_k)$.
  We already know, by Lemma~\ref{lem:Ek}, that for all $s^\abs \in \underline{V}$, if $(s^\abs,a) \in \underline{E}_k$ then 
  all successors $t^\abs$ with $\Delta(s^\abs,a,t^\abs)$ satisfies $t^\abs \in \underline{V}$.
  We are going to show that for all $s^\abs \in \underline{V} \cap \Stat_j^\abs$ with $j\ne k$, for all~$a \in \Act_j$,
  $\Delta^\abs(s^\abs,a,t^\abs)$ implies $t^\abs \in \underline{V}$.

  Consider $s^\abs \in \underline{V}$.
  If~$s^\abs \in \underline{V}_{k,1}$, then for all~$a \in \Act$, $\Delta^\abs(s^\abs,a,t^\abs)$ implies
  that $t^\abs \in \underline{V}_{k,1}$, since $\Agt\setminus\{k\}$ resolves non-determinism.
  The situation is similar if $s^\abs \in \underline{V}_{k,-1}$; for all $a \in \Act_j$, $\Delta^\abs(s^\abs,a,t^\abs)$
  implies ${t}^\abs \in \underline{V}_{k,-1}$. 
  If~$s^\abs \in \underline{V}_{k,0}$, then, by the definition of the outer fixpoint, for all~$a \in \Act_j$,
  $\Delta^\abs(s^\abs,a,{t}^\abs)$ implies that ${t}^\abs \in \underline{V}$.

  Thus $\reach(\AbsA^{\Agt\setminus\{k\}},\sigma) \subseteq \underline{V}$ for all $\sigma \in \strat_k(\underline{E}_k)$.
  It then follows that $\reach(\Game,\gamma(\sigma)) \subseteq \gamma(\underline{V})$.
  So, by Lemma~\ref{lemma:admissible-non-reachable}, and by the fact that $\gamma(\underline{E}_k) \subseteq E_k$, 
  all strategies in $\gamma(\strat_k(\underline{E}_k))$ are value preserving, which is to say,
  belong to $\strat_k(E_k)$.
  
\end{proof}

\paragraph{Help States}
\begin{lemma}
  \label{lemma:approximate-h}
  For all players~$k$,
  \( \gamma(\underline{H}_k) \subseteq H_k \subseteq \gamma(\overline{H}_k)\).
\end{lemma}
\begin{proof}
  Let $s^\abs \in \underline{H}_k$, and let~$a,b \in \Act$ two witnessing actions. 
  For all $s \in \gamma(s^\abs)$, we have $\delta(s,a) \in \gamma(\post_\Delta(s^\abs,a)) \subseteq V_{k,0} \cup V_{k,1}$
  and $\delta(s,b) \in \gamma(\post_\Delta(s^\abs,a)) \subseteq V_{k,0} \cup V_{k,1}$.
  Moreover $\alpha(\delta(s,a)) \in \post_\Delta(s^\abs,a)$, $\alpha(\delta(s,b)) \in \post_\Delta(s^\abs,b)$,
  and $\post_\Delta(s^\abs,a) \cap \post_\Delta(s^\abs,b) = \varnothing$, therefore $\alpha(\delta(s,a)) \ne \alpha(\delta(s,b))$ and thus $\delta(s,a) \neq \delta(s,b)$.
  Hence $s \in H_k$.

  Now, consider any $s \in H_k$; and let $a,b \in \Act$ be such that $\delta(s,a),\delta(s,b) \in V_{k,0} \cup V_{k,1}$
  and $\delta(s,a) \neq \delta(s,b)$.
  If we write $t^\abs = \alpha(\delta(s,a))$ and $u^\abs = \alpha(\delta(s,b))$,
  then $t^\abs,u^\abs \in \overline{V}_{k,0} \cup \overline{V}_{k,1}$,
  and $\Delta(s^\abs,a,t^\abs)$, and $\Delta(s^\abs,b,u^\abs)$; thus $\alpha(s) \in \overline{H}_k$.
  It follows that $H_k \subseteq \gamma(\overline{H}_k)$.
\end{proof}

\paragraph{Synthesizing \AA-winning Strategies}
\begin{lemma}
  \label{lemma:approximate-M}
  We have $\gamma(\underline{M_k'}) \subseteq M_k' \subseteq \gamma(\overline{M_k'})$.
\end{lemma}
\begin{proof}
  We have $\gamma(\phi_k^\abs) = \phi_k$ by assumption on~$\gamma$. Thus, 
  by Lemma~\ref{lemma:approximate-v},
  \[\gamma(( \G\F(\overline{V}_{k,1}) \Rightarrow \phi_k^\abs)) \subseteq 
  \G\F(V_{k,1}) \Rightarrow \phi_k) \subseteq \gamma(( \G\F(\underline{V}_{k,1}) \Rightarrow \phi_k^\abs)).\]
  Similarly, by Lemma~\ref{lemma:approximate-h}, we get 
\(
\gamma\left( \G\F(\overline{V}_{k,0}) \Rightarrow ({\phi}^\abs_k \lor \G\F(\underline{H}_k)) \right)
\subseteq 
\G\F({V}_{k,0} \Rightarrow (\phi_k \lor \G\F(H_k))
\subseteq 
\gamma\left( \G\F(\underline{V}_{k,0}) \Rightarrow ({\phi}^\abs_k \lor \G\F(\overline{H}_k)) \right).\)
It follows that $\gamma(\underline{M'}_k) \subseteq M_k' \subseteq \gamma(\overline{M'}_k)$.


\end{proof}

The following lemma proves Theorem~\ref{thm:abstract}.
\begin{lemma}\label{lemma:winning-under-omega}
  Let $k \in \Agt$ be a player and $\sigma_k$ a strategy of player~$k$.
  If $\sinit^\abs \in \underline{V}$, and $\sigma_k$ is winning for objective $\underline{\Omega'}_k$ in $\AbsA_k'$, 
  then $\gamma(\sigma_k)$ is winning for $\Omega_k'$ in~$\Game'_k$.
\end{lemma}
\begin{proof}
  Let us rewrite
  \[
  \underline{\Omega'}_i = \underline{M'}_i \land \left(\left(\G\F(\Stat^\abs\times\{0\}) \land ( \bigwedge_{j \neq i} \overline{M'}_j \Rightarrow {\phi}^\abs_i ) \right)
   \lor \G\F(\Stat^\abs\times\{\top\})\right).
   \]

  Let~$\sigma_k$ be a winning strategy in~$\AbsA_k'$ for~$\underline{\Omega'}_k$.
  We will show that $\Game', \gamma(\sigma_k) \models \Omega_k'$.

  Consider any run~$\rho$ of~$\Game'_k$ compatible with~$\gamma(\sigma_k)$.
  By definition of $\gamma(\sigma_k)$, $\alpha(\rho)$ is a run of $\AbsA_k'$ compatible with~$\sigma_k$.
  Since $\sigma_k$ is a winning strategy, $\alpha(\rho) \in \underline{M'}_k$, and by Lemma~\ref{lemma:approximate-M} $\rho \in M_k'$.

  We now show that $\rho \in \G\F(\Stat \times \{0,\top\})$.
  By assumption, we have $\AbsA_k', \sigma_k \models \G\F(\Stat^\abs \times \{0,\top\})$,
  which means that for all histories~$h^\abs$ of $\AbsA_k'$ compatible with~$\sigma_k$, 
  $(\last(h^\abs), \sigma(h^\abs)) \in \underline{E}_k$ (otherwise the transition relation of $\AbsA_k'$ would lead  to a $\bot$ state).
  Moreover, since $\sinit^\abs \in \underline{V}$, 
  it follows from Lemma~\ref{lemma:under-admissible-strat} that $(\last(h), \gamma(\sigma)(h)) \in E_k$
  for all histories~$h$ compatible with~$\gamma(\sigma_k)$.  
  Thus no state $(\ast, \bot)$ is reachable under~$\gamma(\sigma)$ in~$\Game'_k$.

  Because of the structure of $\Game'_k$ this means that $\rho$ either visits states of $\Stat\times  \{0\}$ or states of $\Stat\times\{\top\}$ infinitely often:
  \begin{itemize}
    \item If $\rho \in \G\F(\Stat\times \{0\})$, then $\alpha(\rho) \in \G\F(\Stat^\abs \times \{0\})$; so 
      $\alpha(\rho) \in \bigwedge_{j \neq k} \overline{M'}_j \Rightarrow \phi_k^\abs$;
      it follows, by Lemma~\ref{lemma:approximate-M} and the compatibility of the abstraction with players' objectives, that
      $\rho \in \bigwedge_{j \neq k} M'_j \Rightarrow \phi_k$. Thus $\rho \in \Omega_k'$.
    \item
      Otherwise $\rho \in \G\F(\Stat\times\{\top\})$, so $\rho \in \Omega_k'$.
  \end{itemize}
  Thus any outcome $\rho$ of $\gamma(\sigma_k)$ belongs to $\Omega_k'$ which shows it is winning.
\end{proof}

\section{Algorithm for Assume-Guarantee Synthesis}
The assume-guarantee-$\land$ rule was studied in~\cite{CH07} for particular games with three players. However, the given proofs are based on \emph{secure equilibria} which do not actually 
capture assume-guarantee synthesis, so the correctness of the algorithm is not clear.
Here, we give an alternative algorithm for deciding assume-guarantee-$\land$ for multiplayer games, and prove its correctness.

For any game~$\Game$, and state~$s$, we denote by $G_s$ the game obtained making~$s$ the initial state.
Assuming that for each player~$i$ has an objective~$\phi_i$ which is prefix independent, 
let us define $W_i = \{ s \in \Stat \mid \exists \sigma_i.\ G_s , \sigma_i \models \bigwedge_{j\in\Agt \setminus \{i\}} \phi_j \Rightarrow \phi_i \}$.

The following lemma gives a decidable characterization of assume-guarantee  synthesis:
\begin{lemma}\label{lem:assume-guarantee}
  Let~$\phi_i$ be a prefix-independent objective.
  Rule $\AG^\land$ has a solution if, and only if, there is a run~$\rho$ which visits only states of $\bigcap_{i\in\Agt} W_i$ and such that $\rho \models \bigwedge_{i\in\Agt} \phi_i$.
\end{lemma}
\begin{proof}
  \fbox{$\Rightarrow$}
  Let $\sigma_\Agt$ be a solution of $\AG^\land$.
  Let $\rho$ be its outcome.
  We have that $\rho \models \bigwedge_{i\in\Agt} \phi_i$ by hypothesis of
	$\AG^\land$.
  Let $i$ be a player, we show that $\rho$ only visits states of $W_i$.
  This is because $\sigma_i$ is winning for $\bigwedge_{j\in\Agt \setminus \{i\}} \phi_j \Rightarrow \phi_i$.
  For all $k$, $\rho_{\le k}$ is a finite outcome of $\sigma_i$, and the strategy played by $\sigma_i$ after this history is winning for $\bigwedge_{j\in\Agt \setminus \{i\}} \phi_j \Rightarrow \phi_i$, which means that $\rho_{k}$ belongs to $W_i$.
  Hence $\rho$ satisfies the desired conditions.
  
  \fbox{$\Leftarrow$}
  If there is such a run $\rho$, we define the strategy profile $\sigma_\Agt$ to follow this run if no deviation has occurred and otherwise each player $i$ plays a strategy which is winning for $\bigwedge_{j\in\Agt \setminus \{i\}} \phi_j \Rightarrow \phi_i$ if possible.
  We show that such a strategy profile satisfies the assumption of assume-guarantee.
  Obviously $\sigma_\Agt \models \bigwedge_{i\in\Agt} \phi_i$.
  Let $\rho'$ be an outcome of $\sigma_i$ and $k$ the first index such that $\rho'_k \ne \rho_k$.
  The state~$\rho'_{k-1}=\rho_{k-1}$ is not controlled by player~$i$, because $\sigma_i$ follows $\rho$.
  As $\rho_{k-1}$ is in $W_i$ and not controlled by player~$i$, this means that $\rho'_{k} \in W_i$.
  Therefore $\sigma_i$ plays a winning strategy from $\rho'_k$ for the objective $\bigwedge_{j\in\Agt \setminus \{i\}} \phi_j \Rightarrow \phi_i$;
  thus~$\rho'$ satisfies this objective.
  Hence $\sigma_\Agt$ is a solution of $\AG^\land$.
  
\end{proof}

We deduce a polynomial-space algorithm for the $\AG^\land$ rule with Muller objectives:
\begin{theorem}
  For multi-player games with Muller objectives, deciding whether $\AG^\land$
	has a solution is \PSPACE-complete.
\end{theorem}
\begin{proof}
  The algorithms proceed by computing the set $W_i$ for each player $i$ with an algorithm that computes winning regions and then checks whether there is an infinite run in the intersection $\bigcap_{i\in\Agt} W_i$ which satisfies $\bigwedge_{i\in\Agt} \phi_i$.
  This algorithm is correct thanks to Lemma~\ref{lem:assume-guarantee}.

  This is in \PSPACE~because the objective~$\bigwedge_{j\in \Agt\setminus\{i\}} \phi_j \Rightarrow \phi_i$ can be expressed by a Muller condition encoded by a circuit~\cite{dawar13} of polynomial size.
  We can decide in polynomial space if a given state is winning for a Muller condition given by a circuit.
  Thus, the set~$\bigcap_{i \in \Agt} W_i$ can be computed in polynomial space; let us denote by $\Game'$ the game restricted to this set.
  The algorithm then consists in finding a run in $\Game'$ satisfying $\bigwedge_{i\in\Agt} \phi_i$;
  that is, finding a run satisfying a Muller condition, which can be done in polynomial space.
  
\end{proof}

\end{document}